\documentclass[prx,twocolumn,nofootinbib]{revtex4-2}
\usepackage{amsmath,amsfonts,amsthm,mathrsfs,xspace,graphicx,mathtools}
\usepackage{stmaryrd}
\usepackage{dsfont}
\usepackage{nicefrac}
\usepackage{tabularx}
\usepackage[dvipsnames]{xcolor}
\definecolor{linkblue}{HTML}{001487}
\usepackage[colorlinks=true,allcolors=linkblue]{hyperref}
\usepackage{amssymb,latexsym}
\usepackage{tikz}
\usepackage{pgfplots}
\usepackage{url}
\usepackage[UKenglish]{babel}
\usepackage{longfbox}
\usepackage[nameinlink,capitalize,noabbrev]{cleveref}
\usepackage{subcaption}
\crefname{enumi}{Step}{Steps}

\usepackage{enumitem}
\setenumerate[0]{label={\normalfont (\roman*)}}

\newtheorem{theorem}{Theorem}[section]
\newtheorem*{theorem*}{Theorem}

\newtheorem{lemma}[theorem]{Lemma}
\newtheorem{claim}[theorem]{Claim}

\theoremstyle{remark}

\newtheorem{condition}[theorem]{Condition}

\theoremstyle{definition}
\newtheorem{definition}[theorem]{Definition}

\newtheorem{protocol}{Protocol}

\numberwithin{equation}{section}
\newcommand\numberthis{\addtocounter{equation}{1}\tag{\theequation}}

\newcommand{\ket}[1]{|#1\rangle}
\newcommand{\bra}[1]{\langle#1|}
\newcommand{\proj}[1]{\ket{#1}\!\bra{#1}}
\newcommand{\projs}[1]{\ket{#1}\!\bra{~}}
\DeclarePairedDelimiterX\braket[2]{\langle}{\rangle}{#1 \delimsize\vert #2}

\newcommand{\tr}[1]{\mbox{\rm Tr}\!\left[ #1 \right]}
\newcommand{\ptr}[2]{\mbox{\rm Tr}_{#1}\!\left[ #2 \right]}
\newcommand{\pr}[1]{{\rm Pr}\!\left[ #1 \right]}
\newcommand{\prs}[2]{{\rm Pr}_{#1}\!\left[ #2 \right]}
\newcommand{\1}{\mathds{1}}

\newcommand{\N}{\ensuremath{\mathbb{N}}}
\newcommand{\R}{\ensuremath{\mathbb{R}}}

\let\H\relax
\newcommand{\H}{\mathcal{H}}
\DeclareMathOperator*{\E}{\mathbb{E}}

\newcommand{\ot}{\ensuremath{\otimes}}
\newcommand{\deq}{\coloneqq}

\newcommand{\cA}{\ensuremath{\mathcal{A}}}
\newcommand{\cB}{\ensuremath{\mathcal{B}}}
\newcommand{\cC}{\ensuremath{\mathcal{C}}}

\newcommand{\cF}{\ensuremath{\mathcal{F}}}

\newcommand{\cH}{\ensuremath{\mathcal{H}}}
\newcommand{\cI}{\ensuremath{\mathcal{I}}}

\newcommand{\cM}{\ensuremath{\mathcal{M}}}
\newcommand{\cN}{\ensuremath{\mathcal{N}}}

\newcommand{\cP}{\ensuremath{\mathcal{P}}}

\newcommand{\cR}{\ensuremath{\mathcal{R}}}
\newcommand{\cS}{\ensuremath{\mathcal{S}}}
\newcommand{\cT}{\ensuremath{\mathcal{T}}}
\newcommand{\cU}{\ensuremath{\mathcal{U}}}
\newcommand{\cV}{\ensuremath{\mathcal{V}}}
\newcommand{\cX}{\ensuremath{\mathcal{X}}}
\newcommand{\cY}{\ensuremath{\mathcal{Y}}}
\newcommand{\cZ}{\ensuremath{\mathcal{Z}}}

\newcommand{\mbP}{\ensuremath{\mathbb{P}}}

\newcommand*{\freq}[1]{\mathsf{freq}(#1)}
\newcommand*{\Var}[1]{\mathsf{Var}(#1)}

\newcommand*{\Max}[1]{\mathsf{Max}(#1)}
\newcommand*{\MinSigma}[1]{\mathsf{Min}_{\Sigma}(#1)}
\newcommand*{\Min}[1]{\mathsf{Min}(#1)}

\newcommand{\norm}[1]{\left\lVert#1\right\rVert}

\newcommand{\setft}[1]{\textnormal{#1}}
\newcommand{\pos}{\setft{Pos}}
\newcommand{\states}{\setft{S}}

\newcommand{\bits}{\ensuremath{\{0, 1\}}}

\newcommand{\eps}{\varepsilon}

\newcommand{\ext}{\textnormal{\textsc{Ext}}}
\newcommand{\hash}{\textnormal{\textsc{Hash}}}
\newcommand{\ec}{\textnormal{\textsc{ec}}}
\newcommand{\kv}{\textnormal{\textsc{kv}}}
\newcommand{\pa}{\textnormal{\textsc{pa}}}
\newcommand{\ca}{\textnormal{\textsc{ca}}}
\newcommand{\pd}{\textnormal{\textsc{pd}}}
\newcommand{\rk}{\textnormal{\textsc{rk}}}
\newcommand{\ev}{\textnormal{\textsc{ev}}}
\newcommand{\err}{\texttt{err}}

\newcommand{\sth}{{\rm~s.t.~}}
\newcommand{\tand}{\textnormal{~and~}}

\let\eps\varepsilon

\newcommand{\hmin}{H_\setft{min}}
\newcommand{\hmax}{H_{\ensuremath{\text{max}}}}
\newcommand{\dnormal}[2]{D \! \left( #1 \, \middle\Vert \, #2 \right)}

\newcommand{\hau}{H^{\shortuparrow}_\alpha}

\newcommand{\cptp}{\setft{CPTP}}

\newcommand{\id}{\setft{id}}

\newcommand{\Z}{\mathsf{Z}}
\newcommand{\X}{\mathsf{X}}
\newcommand{\zx}{\{\Z, \X\}}

\tikzstyle{porte} = [fill=gray!25, draw]

\usepackage{appendix}
\Crefname{appsec}{Supplementary Note}{Supplementary Notes}

\makeatletter 
    
\renewcommand\onecolumngrid{
\do@columngrid{one}{\@ne}%
\def\set@footnotewidth{\onecolumngrid}
\def\footnoterule{\kern-6pt\hrule width 1.5in\kern6pt}%
}
\makeatother

\begin{document}

\title{Security of quantum key distribution from generalised entropy accumulation}

\author{Tony Metger}
\email{tmetger@ethz.ch}

\author{Renato Renner}
\affiliation{Institute for Theoretical Physics, ETH Z{\"u}rich, 8093 Z{\"u}rich, Switzerland}


\begin{abstract}
The goal of quantum key distribution (QKD) is to establish a secure key between two parties connected by an insecure quantum channel. To use a QKD protocol in practice, one has to prove that a finite size key is secure against \textit{general attacks}: no matter the adversary's attack, they cannot gain useful information about the key. A much simpler task is to prove security against \textit{collective attacks}, where the adversary is assumed to behave identically and independently in each round. 
In this work, we provide a formal framework for general QKD protocols and show that for any protocol that can be expressed in this framework, security against general attacks reduces to security against collective attacks, which in turn reduces to a numerical computation.
Our proof relies on a recently developed information-theoretic tool called \textit{generalised entropy accumulation} and can handle generic \textit{prepare-and-measure protocols} directly without switching to an entanglement-based version.
\end{abstract}

\maketitle

\section{Introduction} \label{sec:intro} Quantum key distribution (QKD) considers the following scenario: two parties, Alice and Bob, can communicate via an insecure quantum channel and an authenticated classical channel.\footnote{An insecure quantum channel allows the adversary to intercept and tamper with any quantum state sent across the channel. An authentic classical channel is one where an adversary can read every message sent across the channel, but cannot impersonate either party; for example, the adversary cannot convince Bob that a certain message was sent by Alice when in fact, it was not.} 
Using these resources, Alice and Bob would like to establish a secure shared key, i.e.~a piece of information that is known to both of them, but entirely unknown to an adversary Eve~\cite{bb84, e91}.

The key difficulty in establishing the security of a QKD protocol is that one has to take into consideration any possible attack that the adversary Eve may perform. 
For example, in one round of the protocol, Eve may gather a piece of quantum side information about the quantum state sent via the insecure channel. 
This piece of side information could be combined with side information from previous rounds to plan Eve's attack for the next round, resulting in a very complicated multi-round attack.
Additionally, Alice and Bob can only execute a certain finite number of rounds, introducing statistical finite-size effects.
A security proof that takes both of these challenges into account is called a \textit{finite-size security proof against general attacks} (also referred to as \textit{coherent attacks})~\cite{bruss_finite,scarani_renner_finite,Cai_2009}.
Such a proof is required to safely deploy a QKD protocol in practice.

Due to the difficulty of proving finite-size security against general attacks, many protocols are first analysed for \textit{collective attacks}, for which very general numerical techniques have been developed (see e.g.~\cite{coles2016numerical, Winick2018reliablenumerical, wang2019characterising, primaatmaja2019versatile, brown2021device, brown2021computing, tan2021computing, hu2021robust, araujo2022quantum}). 
For a security proof against collective attacks one makes the assumption that Alice and Bob execute infinitely many rounds of the protocol and Eve behaves independently and identically in each round. 
This is also called the \textit{i.i.d.~asymptotic setting}.
These assumptions are of course unrealistic, but a collective attack proof is a useful theoretical tool as it can often be converted into a finite-size proof against general attacks using techniques such as the quantum de Finetti theorem~\cite{renner_thesis,christandl2009postselection}. 
Such general techniques are very powerful, but typically require additional assumptions on the protocol and can significantly lower the amount of key that can be extracted compared to the collective attack scenario.

In this work, we show that  security against collective attacks implies finite-size security against general attacks for a broad class of protocols. 
The main feature of our security proof is its generality: while many existing security proofs work well for particular protocols, our approach works for any generic protocol satisfying a few structural assumptions.
Furthermore, it provides a natural way of proving security against general attacks, with the proof being in close correspondence to the structure of the original protocol, whereas previous techniques often required the protocol to be transformed into a theoretically equivalent one to fit into the framework of a particular proof technique. 
In particular, our technique can be applied directly to \textit{prepare-and-measure protocols} without transforming them into an entanglement-based version.
Furthermore, our technique provides bounds that are independent of the dimension of the underlying Hilbert space; instead, the bound depends only on the number of possible classical outputs that Alice and Bob may receive. 
This is particularly important for photonic QKD protocols, where the underlying Hilbert space is a Fock space with unbounded dimension~\cite{lutkenhaus_pra,inamori2007unconditional}, and is also useful for (semi-)device-independent protocols.
We give a more detailed comparison between our technique and previous ones below.

Our main result, \cref{thm:general_qkd_pm}, reduces the task of showing security against general attacks to the (generally much simpler) task of proving security against collective attacks.
In \cref{sec:pm_general}, we present this result by providing a very general template prepare-and-measure QKD protocol (\cref{prot:qkd-pm}) and proving its finite-size security against general attacks assuming a security statement for collective attacks.
Furthermore, we adapt existing numerical techniques~\cite{coles2016numerical,Winick2018reliablenumerical} to show that a collective attack bound for any instance of this general protocol can be computed by solving a convex optimisation problem (\cref{sec:col_attack_bounds}).
Hence, to show security for a particular protocol of interest, one can simply express this protocol as an instance of our general template protocol, numerically determine a collective attack bound, and deduce security against general attacks immediately from our \cref{thm:general_qkd_pm}.
We illustrate the procedure of proving security using our framework by applying it to the B92 protocol in \cref{sec:b92};
this yields the first finite-size security proof against general attacks for the B92 protocol that converges to the optimal key rate in the asymptotic limit of infinitely many rounds.
We also note that the proof technique underlying the security proof of our template protocol is very general and can easily be adapted to protocols that do not fit precisely into our framework.
For example, one could incorporate advantage distillation~\cite{maurer_advantage,gottesman_two_way} by applying our technique in a block-wise rather than a round-wise manner.

Our proof employs a recent information-thereotic result, called the \emph{generalised entropy accumulation theorem} (GEAT)~\cite{gen_eat}.
This theorem is an information-theoretic statement about the min-entropy produced by sequential quantum processes; see \cref{sec:geat_prelim} for a more detailed explanation.
For our security proof, we show that generic QKD protocols can be expressed in such a form that the GEAT provides a tight lower-bound on a certain min-entropy associated with the protocol, which implies security of the protocol against \textit{general} attacks.
This entropic lower bound depends on the security statement against \textit{collective} attacks; it is in this sense that our proof reduces security against general attacks to security against collective attacks.

There are a number of existing techniques for converting security proofs against collective attacks into ones against general attacks.
The most widely used ones are either based on the quantum de Finetti theorem~\cite{renner_thesis} (or the related post-selection technique~\cite{christandl2009postselection}), or they employ the entropy accumulation theorem (EAT)~\cite{eat}, of which the GEAT is a generalisation.
We briefly describe how each of these relate to our technique using the GEAT.

The quantum de Finetti theorem and related methods such as the post-selection technique rely on the permutation-symmetry between different rounds of the protocol to reduce general to collective attacks.
While not every protocol possesses this permutation symmetry naturally, it can usually be enforced by including an additional ``symmetrisation step'' in the protocol.
The main downside of these techniques is that the bounds they achieve scale unfavourably with the dimension of the underlying Hilbert space, i.e.~the Hilbert space that contains the states sent from Alice to Bob.
This means that these techniques only yield useful bounds for protocols with a small Hilbert space dimension, e.g.~the BB84 or B92 protocols~\cite{bb84, b92}.
However, practical implementations of QKD protocols do not always satisfy this requirement; for example, many protocols use laser pulses as the means by which Alice sends a quantum state to Bob~\cite{qkd_scarani_review,practical_qkd_review}, and such laser pulses are described in a Fock space whose dimension is in principle unbounded.
While methods for truncating the Fock space have been developed~\cite{renner2009finetti,beaudry_squashing}, this introduces additional complications and may lead to weak bounds if the dimension of the truncated Fock space remains large.

In contrast, the EAT and GEAT provide bounds that do not depend on the dimension of the underlying Hilbert space.
In fact, the EAT and GEAT share (almost) the same second-order terms, so one does not incur a noticeable loss in parameters when using the GEAT compared to the EAT. 
This dimension-independence of the second-order terms means that the EAT can also be used to prove security for device-independent or semi-device-independent proto\-cols~\cite{arnon2019simple}.
For an example of a device-independent protocol that can be treated with the GEAT, but not the EAT, see Ref.~\cite{gen_eat}.

The main difference between the EAT and the GEAT is that the EAT deals with a more restrictive model of how side information can be generated during the protocol.
The  GEAT allows Eve's side information to be updated in an arbitrary way.
In contrast, the EAT requires that new side information must be output in a round-by-round manner subject to a Markov condition between rounds, and once side information has been output it cannot be updated anymore.
In general, it is not possible to model the way that Eve actively intercepts quantum states and updates her side information in a prepare-and-measure protocol by a process that outputs side information in a round-by-round manner subject to the Markov condition.
As a consequence, the EAT cannot ``naturally'' deal with general prepare-and-measure protocols.
Instead, one first has to convert a prepare-and-measure protocol into an entanglement-based protocol.
This can be done as follows: if Alice prepares one of a set of pure states $\{\ket{\psi^j}_Q\}_j$ with probability $p(j)$ and stores the index $j$ specifying the state in her register $A$, we can replace this by Alice preparing a state $\ket{\tilde \psi}_{AQ} = \sum_{j} \sqrt{p(j)} \ket{j}_A \ket{\psi^j}_Q$ and later measuring her system $A$.
Then, we can model Eve's attack by replacing this state $\ket{\tilde \psi}_{AQ}$ by an arbitrary state $\ket{\hat \psi}_{AQE}$ prepared by Eve, subject to the constraint that Alice's marginal, which Eve cannot access in the prepare-and-measure protocol, is ``correct'', i.e.~$\tilde \psi_A = \hat \psi_A$.
This additional constraint is an artificial one in the sense that it is not something that Alice and Bob check in the actual protocol, and it is unclear how it can be incorporated into a security proof using the EAT in a natural way.
As a result, it appears difficult or impossible to use the EAT to obtain reasonable finite-size key rates for prepare-and-measure protocols except in very simple cases.
In contrast, the GEAT is able to deal with prepare-and-measure protocols directly, circumventing this issue entirely.
The GEAT's ability to deal with prepare-and-measure protocols without any dependence on the dimension of the underlying Hilbert space makes it particularly useful for photonic prepare-and-measure protocols, which are of  practical interest.

In addition to these general techniques for reducing security against general attacks to security against collective attacks, there are also more specialised techniques that directly prove security against general attacks without an explicit reduction to collective attacks.
Perhaps the most common of these are phase-error correction and entropic uncertainty techniques, both of which use the complementarity of different measurements in the protocol as the starting point for a security proof (see e.g.~\cite{lo-chau,shor-preskill,tamaki2003unconditionally,three-state,Koashi_2009,tomamichel2012tight}).
These security proofs usually give very tight bounds for ``symmetric'' protocols (i.e.~protocols relying on mutually unbiased measurement bases, even though these bases need not be chosen with equal probability) where they can be applied naturally, and can also be extended to symmetric protocols with experimental imperfections that slightly break the symmetry, e.g.~using the reference state technique~\cite{pereira2020quantum,pereira2022modified}.
In addition, various other proof techniques that use the symmetry of specific protocols have been developed (see e.g.~\cite{christandl2004generic, rgk_infotheoretic, akmb_renyi}).
Our general security proof is not meant to replace more specialised techniques in cases where the protocol is symmetric enough for these to be applied; instead, we provide a unifying framework that can also prove the security of protocols for which no specialised techniques are available, or for which such techniques do not yield asymptotically tight bounds.
An example of this is the B92 protocol, for which we give a full security proof in~\cref{sec:b92}: for this protocol, complementarity-based methods have been used to prove security against general attacks~\cite{tamaki2003unconditionally,tamaki2004unconditional}, but the resulting key rates do not converge to the collective attack rate in the limit of infinitely many rounds.
In contrast, a direct application of our general framework yields the first asymptotically tight finite-size security proof against general attacks for the B92 protocol.

\section{Results} \label{sec:pm_general}

\subsection{Framework for prepare-and-measure protocols} \label{sec:gen_pm_protocol}

Our main result, \cref{thm:general_qkd_pm}, shows that for a broad class of prepare-and-measure protocols, security against collective attacks implies security against general attacks.
To make this result easy to use, we phrase it as a security statement for a general ``template protocol''; many existing prepare-and-measure protocols can be viewed as an instance of this template protocol, and their security then follows from the security of the general template protocol.
For protocols that do not fit exactly into this template, the security proof can usually easily be adapted from our proof of \cref{thm:general_qkd_pm}.

\begin{figure*}
\begin{longfbox}[breakable=false, padding=1em, padding-right=1.8em, padding-top=1.2em, margin-top=1em, margin-bottom=1em]
\begin{protocol} {\bf General prepare-and-measure QKD protocol} \label{prot:qkd-pm} \end{protocol}
\noindent\underline{Protocol arguments}  \vspace{-0.8ex}
\vspace{-0.6ex}
\begin{center}
\begin{tabularx}{\textwidth}{r c X}
$n \in \N$ &:& number of rounds \\
$\psi_{UQ}$ &:& quantum state prepared by Alice, where $U$ is classical with alphabet $\cU$ and $Q$ is quantum \\
$\{N^{(v)}\}_{v \in \cV}$ &:& POVM acting on Hilbert space $\H_Q$ describing Bob's trusted measurements (where $\cV$ is some finite set of possible outcomes) \\
$\pd: \cU \times \cV \to \cI$ &:& function describing transcript of public discussion (where $\cI$ is some finite alphabet) \\
$\rk: \cU \times \cI \to \cS$ &:& function describing Alice's raw key generation (where $\cS$ is the alphabet of the raw key) \\
$\ev: \cV \times \cI \times \cS \to \cC$ &:& function ``evaluating'' each round by assigning a label from the alphabet $\cC$ \\ 
$\lambda_{\ec} \in \N_0$ &:& length of bit string exchanged during error correction step \\ 
$k_{\ca} > 0$ &:& required amount of single-round entropy generation \\
$\eps_{\kv}, \eps_{\pa} > 0$ &:& tolerated errors during key validation and privacy amplification steps \\
$\ca: \mbP(\cC) \to \R$ &:& affine function corresponding to collective attack bound \\
$l \in \N$ &:& length of final key \\
\end{tabularx}
\end{center}

\vspace{0.1ex}

\noindent\underline{Protocol steps}
\begin{enumerate}[label=(\arabic*)]
    \setlength{\itemsep}{0.4ex}
    \setlength{\parskip}{0.4ex}
    \setlength{\parsep}{0.4ex} 
\item \emph{Data generation.} Alice prepares $\psi_{U^n Q^n} = \psi_{UQ}^{\ot n}$ and sequentially sends the systems $Q_1, \dots, Q_n$ to Bob via a public quantum channel.
For each $i \in \{1, \dots, n\}$, Bob measures $\{N^{(v)}\}_{v \in \cV}$ on register $Q_i$ and records the outcome in register $V_i$. \label{step_pm:data_gen}
\item \emph{Public discussion.}
For each $i \in \{1, \dots, n\}$, Alice and Bob publicly exchange information $I_i = \pd(U_i, V_i)$. \label{step_pm:pd}
\item \emph{Raw key generation.} For each $i \in \{1, \dots, n\}$, Alice computes $S_i = \rk(U_i, I_i)$. \label{step_pm:raw_key_gen}
\item \emph{Error correction.} Alice and Bob publicly exchange information $\ec \in \bits^{\lambda_{\ec}}$, which can depend on $U^n, V^n$, and $I^n$. Bob computes $\hat S^n(\ec, V^n, I^n) \in \cS^n$. \label{step_pm:ec}
\item \emph{Raw key validation.} Alice chooses a function $\hash: \cS^n \to \bits^{\lceil \log(1/\eps_{\kv}) \rceil}$ from a universal hash family $\cF$ (\cref{def:univ_hash}) according to the associated probability distribution $P_{\cF}$ and publishes a description of $\hash$ and the value $\hash(S^n)$. \label{step_pm:rkv}
Bob computes $\hash(\hat S^n)$ and aborts the protocol if $\hash(S^n) \neq \hash(\hat S^n)$.
\item \emph{Statistical check.} For each $i \in \{1, \dots, n\}$, Bob sets $\hat C_i = \ev(V_i, I_i, \hat S_i)$. Bob then computes $k = \ca(\freq{\hat C^n})$. If $k < k_\ca$, he aborts the protocol. \label{step_pm:stat}
\item \emph{Privacy amplification.} Alice and Bob convert their registers $S^n$ and $\hat S^n$ to a binary representation, obtaining strings of length $m$.
Alice chooses a seed $\mu \in \bits^m$ uniformly at random and publishes her choice.
Alice and Bob compute $l$-bit strings $K = \ext(S^n, \mu)$ and $\hat K = \ext(\hat S^n, \mu)$, respectively, where $\ext: \bits^m \times \bits^m \to \bits^l$ is a quantum-proof strong $(l + \lceil 2 \log(1/\eps_{\pa})\rceil, \eps_{\pa})$-extractor (\cref{def:extractor}). \label{step_pm:pa}
\end{enumerate} 
\end{longfbox}
\end{figure*}

Our template protocol is described formally in \cref{prot:qkd-pm}; here, we make a few additional remarks regarding this general protocol, using the notation introduced in \cref{prot:qkd-pm}.
Firstly, without loss of generality we can assume that the cq-state $\psi_{UQ}$ is of the form $\psi_{UQ} = \sum_{u} p(u) \proj{u} \ot \proj{\psi}_{Q|u}$ for a probability distribution $p(u)$ and \emph{pure} states $\proj{\psi}_{Q|u}$.
This means that Alice chooses a value $u$ according to $p(u)$ and then sends the pure state $\proj{\psi}_{Q|u}$ to Bob.
The reason that we can assume that $\proj{\psi}_{Q|u}$ is pure is that if Alice wanted to send a mixed state, she could express that mixed state as a mixture of pure states, send one of those pure states, and later ``forget'' which of the pure states she sent as part of the map $\rk$.

Secondly, in \cref{prot:qkd-pm} Bob measures a POVM $\{N^{(v)}\}$ with outcomes $v \in \cV$.
More commonly, we think of Bob as choosing an input $y$ according to some distribution $q(y)$ and receiving an output $b \in \cB$.
This can be described by a collection of POVMs $\{\tilde N^{(b)}_y\}_{b \in \cB}$, one for each possible input $y$.
For example, Bob might choose uniformly at random whether to measure a qubit in the computational or Hadamard basis.
In that case, $y$ would be the basis choice, and for each $y$, $\{\tilde N^{(b)}_y\}_{b \in \cB}$ is the measurement in the chosen basis.
However, since Bob's measurements are trusted, the distinction between inputs and outputs is unnecessary: we can convert a set of POVMs $\{\tilde N^{(b)}_y\}_{b \in \cB}$ with an input distribution $q(y)$ into an equivalent single POVM $\{N^{(v)}\}_{v \in\cV}$ by choosing $\cV = \cY \times \cB$ and $N^{(y, b)} = q(y) \tilde N^{(b)}_y$.
This satisfies the required property of a POVM: 
\begin{align*} 
\sum_{y, b} N^{(y, b)} = \sum_{y} q(y) \sum_{b} \tilde N^{(b)}_y = \sum_y q(y) \1 = \1 \,,
\end{align*}
where we used the fact that $\{\tilde N^{(b)}_y\}_{b \in \cB}$ is a POVM for the first equality and the fact that $q(y)$ is a probability distribution for the second.
One can think of $N^{(y, b)}$ as first choosing $y \in \cY$ according to $q(y)$ and then measuring $\{\tilde N^{(b)}_y\}$ on the state, providing $(y, b)$ as output.

Thirdly, the function $\pd$ describes the total information exchanged during the public discussion (\cref{step_pm:pd}) for one round $i$ of the protocol.
The details of how the public discussion takes place are of no concern to the protocol: in general, Alice and Bob may exchange multiple rounds of back-and-forth communication during this step, and $\pd$ describes the transcript of the entire exchange.
For example, in a protocol that includes a sifting step, the public discussion would include the information necessary to decide which rounds to sift out; the actual sifting would occur in the raw key generation step, where Alice's function $\rk$ can use the information from the public discussion to put a special symbol (e.g.~$\bot$) as the raw key for rounds that are sifted out.

Additionally, the protocol distinguishes between information $I_i$ published during \cref{step_pm:pd} and error correction information $\ec$ published during \cref{step_pm:ec}.
The difference between these two steps is that $I_i$ may only depend on the inputs $U_i$ and $V_i$ generated during the $i$-th round of measurements.
This means that $I_i$ is generated in a round-by-round manner and will enter in the single-round security statement (or collective attack bound, see \cref{def:coll_attack}).
In contrast, $\ec$ is \emph{global} information of a fixed length $\lambda_\ec$, i.e.~it can depend arbitrarily on information generated during all rounds of the protocol, but to obtain a good key rate, $\lambda_{\ec}$ should be as short as possible.\footnote{The bound on the length of $\ec$ is needed in \cref{eqn:sub_lambdaec}, where we use it to remove the error correction information from the conditioning system.
We note that one can replace \cref{eqn:sub_lambdaec} by a slightly more sophisticated chain rule that subtracts a (one-shot) mutual information between $\ec$ and $S^n$.
In that case, the protocol needs to specify an upper bound on this mutual information instead of the length $\lambda_\ec$.}

Finally, we note that in \cref{prot:qkd-pm}, Alice and Bob first perform error correction, and afterwards Bob uses his error-corrected guess for Alice's raw key for the purposes of the statistical check.
An alternative that is commonly used in existing QKD protocols is that Alice and Bob publish part of their data in a separate \textit{parameter estimation step} before the error correction step and use this public information to run a statistical check.
Our \cref{prot:qkd-pm} can easily be modified to include protocols of this form.\footnote{For the modified protocol, the security proof stays exactly the same, except that the reduction from \cref{thm:general_qkd_pm} to \cref{claim:eat_qkd_pm} now follows almost trivially and does not need the argument from \cref{app:entr_reduction}.}

\textbf{Example: BB84 protocol as an instance of \cref{prot:qkd-pm}.}
To gain further intuition for \cref{prot:qkd-pm}, we describe how to reproduce the well-known BB84 protocol as an instance of our general \cref{prot:qkd-pm}.
In the BB84 protocol, Alice sends a random state from the set $\{\ket{0}, \ket{1}, \ket{+}, \ket{-}\}$, where $\ket{\pm} = \frac{\ket{0} \pm \ket{1}}{\sqrt{2}}$ are the Hadamard basis states.
As her information $U_i$, Alice records which state she sent, i.e.~she records the basis $x \in \bits$ and the value $a \in \bits$.
Hence, for the BB84 protocol, 
\begin{align*}
\psi_{UQ} = \frac{1}{4} \sum_{x, a \in \bits} \proj{x, a}_{U} \ot H^x \proj{a}_Q H^x \,,
\end{align*}
where $H$ is the Hadamard gate and $H^0 = \id, H^1 = H$.
Bob's measurements output a basis choice $y \in \bits$ and the outcome $b$ of a single-qubit measurement in that basis (with $y = 0$ corresponding to the computational and $y=1$ to the Hadamard basis).
Therefore, his measurements are described by a POVM on system $Q$ consisting of elements 
\begin{align*}
N^{(y, b)} = \frac{1}{2} H^y \proj{b} H^y \,.
\end{align*}

During the public discussion phase, Alice and Bob publish their basis choices $x_i$ and $y_i$ for each of the rounds. Therefore, for $U_i = (x_i, a_i)$ and $V_i = (y_i, b_i)$, 
\begin{align*}
I_i = \pd(U_i, V_i) = (x_i, y_i) \,.
\end{align*}

To generate her raw key, for each round Alice checks whether the basis choices $x_i$ and $y_i$ are the same: if so, she uses her measurement outcome $a_i$ for the raw key, and otherwise she discards that round. Formally, 
\begin{align*}
S_i = \rk(U_i, I_i) = \rk((x_i, a_i), (x_i,  y_i)) = 
\begin{cases}
a_i & \textnormal{if $x_i = y_i$,} \\
\bot & \textnormal{otherwise.}
\end{cases}
\end{align*}

Finally, for the statistical check in \cref{step_pm:stat}, Bob checks whether his guess $\hat S^n$ for Alice's string matches his own raw data.
In fact, Bob can only do this check on a small subset of indices $i$.
The reason is that for our definition of collective attack bounds (\cref{def:coll_attack}) and the security proof (\cref{thm:general_qkd_pm}), we are bounding the entropy conditioned on the systems $C^n$, i.e.~we are essentially assuming that all of the statistical information gets leaked to Eve.\footnote{In practice, Bob can of course keep the statistical information $C^n$ private and Eve would only learn the (single-bit) abort decision. However, conditioning on $C^n$ is necessary in the proof of \cref{claim:eat_qkd_pm} in order to satisfy \cref{eqn:measurement_condition}. It is an interesting open question to see whether the GEAT and our proof can be improved to get rid of the conditioning on $C^n$. In that case, Bob could base his abort decision on data $V_i$, $I_i$, and $\hat S_i$ from \emph{all} rounds. However, this would only provide a relatively minor improvement to the key rate as Bob can already make a well-informed abort decision based on a small fraction of rounds.}
Hence, Bob chooses a value $T_i$ at random with $\pr{T_i = 1} = \gamma$ (where $\gamma$ is the \emph{testing probability}, and the choice of $T_i$ can formally be included into $V_i$), and then sets
\begin{align*}
\hat C_i &= \ev(V_i, I_i, \hat S_i) = \ev((y_i, b_i), (x_i, y_i), \hat S_i) \\
&= \begin{cases}
\bot & \textnormal{if $x_i \neq y_i$ or $T_i = 0$,} \\
1 & \textnormal{if $x_i = y_i$, $T_i = 1$, and $b_i = \hat S_i$,}\\
0 & \textnormal{otherwise.}
\end{cases}
\end{align*}
Intuitively, $\bot$ denotes that no useful check can be performed in this round, ``1'' means the check has passed, and ``0'' means the check has failed.

\subsection{Modelling Eve's attack} \label{sec:eve_attack}
In \cref{prot:qkd-pm}, Eve can obtain information about the final key $K$ in two ways: firstly, Eve can observe the classical information published by Alice and Bob during the protocol, e.g.~the error correction information $\ec$.
In a security proof, this is easy to handle, as Alice and Bob have full control over what information they publish.
Secondly, Eve can intercept the quantum systems $Q_i$ sent from Alice to Bob in \cref{step_pm:data_gen}.
This is much harder to analyse in a security proof as Eve can perform arbitrary operations on the systems $Q_i$ and we need to bound the amount of information Eve can gain about Alice's and Bob's raw key from tampering with the systems $Q_i$ without being detected.
The set of actions Eve performs on the systems $Q_i$ is called Eve's attack.

In principle, Eve could collect all of the $n$ systems $Q_1, \dots, Q_n$, perform an arbitrary quantum channel $\cA: Q^n \to E Q^n$, and send the output on systems $Q^n$ to Bob.
The system $E$ would be kept by Eve and would contain her (potentially quantum) side information about the final key.

To analyse the security of a prepare-and-measure protocol with the GEAT, we need to introduce an extra condition.
\begin{condition} \label{cond:eve_seq}
Eve can only be in possession of one of the systems $Q_i$ at the same time.
\end{condition}
Since Alice sends the systems $Q_1, \dots, Q_n$ sequentially in \cref{step_pm:data_gen}, this means that with this additional condition, Eve's most general attack also takes a sequential form.
More formally, with this condition, the most general attack Eve can perform is described by a sequence of maps $\cA_i: E'_{i-1} Q_i \to E'_i Q_i$, where $E'_i$ are arbitrary quantum systems that contain Eve's side information after having intercepted the $i$-th system $Q_i$.
(The system $E_0$ can be chosen to be trivial without loss of generality, but we will not need this for our security proof.)

In fact, it is easy for Alice and Bob to enforce \cref{cond:eve_seq} by checking that system $Q_i$ has arrived on Bob's side before $Q_{i+1}$ is sent.
The downside of this simple strategy is that if Alice and Bob are far apart, it limits the number of signals that can be sent per unit time.

To circumvent this, Alice and Bob can agree on a ``schedule'' on which signals are transmitted, i.e.~they decide when Alice will send out each signal, so Bob, being aware of its travel time without Eve's interference, knows when to expect to receive it.
Then, assuming that Eve cannot significantly speed up the transmission of signals, this would ensure that \cref{cond:eve_seq} is satisfied without Alice having to wait for Bob's confirmation to send the next signal (see \cref{fig:spacetime} for an illustration of this).
Whether or not the assumption that Eve cannot significantly speed up the transmission of signals is realistic depends on the specific QKD setup: for example, if signals are transmitted from Alice to Bob through vacuum (e.g.~in satellite-to-satellite QKD), they travel at the speed of light and cannot be sped up further by Eve, so~\cref{cond:eve_seq} can be enforced by sending signals on a pre-agreed schedule without issues.

On the other hand, if Alice and Bob exchange signals via a (very long) optical fiber, Eve could in principle extract the signal at the start of the fiber, transmit it through free space, and then re-insert it into the fiber on Bob's side. Since the speed of light in a fiber is slower than in free space, this would enable Eve to have simultaneous access to a (relatively small) set of $s$ sped-up signals, perform some attack involving this set of signals, and then feed the ``first'' of these signals to Bob in such a way that it arrives at the time expected by Bob;
then, Eve could add the next sped-up signal to her set, apply another attack to that set of $s$ signals, and so on.
Such an attack would violate \cref{cond:eve_seq}, but it would go unnoticed by Alice and Bob since the signals do arrive at the expected times on Bob's end.

Setting aside the question of how realistic it is for Eve to perform such an attack, this issue can be addressed by relaxing \cref{cond:eve_seq} so that instead of requiring Eve to be in possession of only one signal at a time, we allow her to be in possession of $s$ signals at a time.
To prove security under this weakened condition, we can divide the signals into interleaved groups such that any two signals within a group are $s$ rounds apart, use a standard chain rule for min-entropies (or Renyi entropies) to divide the total entropy into a sum of group-wise entropies, and simply apply our analysis at the level of these groups.
Our proof then goes through essentially unchanged, although the resulting second-order terms in the key rate will depend on the allowed number $s$ of signals available to Eve at a time.
We explain this modification in more detail in \cref{app:block_seq} and focus on the case where \cref{cond:eve_seq} holds exactly in the main text for simplicity.

We have now seen how to model Eve's general attack under \cref{cond:eve_seq}.
In contrast to such general sequential attacks, collective attacks only allow Eve to perform the same independent attack in each round of the protocol. 
Hence, a collective attack can be modelled by a map $\cA: Q \to E Q$, which Eve applies in each round of the protocol, so Eve's full attack over $n$ rounds is given by the tensor product map $\cA^{\ot n}: Q^n \to E^n Q^n$.
Proving security against this restricted class of attacks is typically much easier than proving security against general attacks.
However, we stress that, unlike \cref{cond:eve_seq}, the assumption that Eve performs only a collective attack cannot be enforced by Alice and Bob.
Therefore, a security proof that only considers collective attacks is insufficient for practical applications.


\subsection{Collective attack bounds} \label{sec:col_attack_bounds}
If one restricts Eve to performing collective attacks, it is known that in the limit $n \to \infty$ of many rounds the key rate is given by a simple entropic expression that only involves quantities corresponding to a single round of the protocol~\cite{devetak2005distillation}.\footnote{The entropic expression for the key rate in~\cite{devetak2005distillation} already includes information leaked to Eve during the error correction step assuming an optimal error correcting protocol. Our \cref{def:coll_attack} does not include a term corresponding to this -- instead in \cref{prot:qkd-pm} we assume that the error correction information has length at most $\lambda_{\ec}$, which we can later subtract from the length of the final key that can be generated.}
More formally, we can view a collective attack bound as a map that takes as input the statistics corresponding to a single round of the protocol and outputs a lower bound on a certain conditional entropy, which specifies how much key can safely be extracted from a state with those statistics.

\begin{definition}[Collective attack bound for \cref{prot:qkd-pm}] \label{def:coll_attack}
Fix arguments $\psi_{UQ}$, $\{N^{(v)}\}_{v \in \cV}$, $\pd$, $\rk$, and $\ev$ for \cref{prot:qkd-pm}.
Suppose that Alice and Bob run a single round (i.e.~$n=1$) of \cref{prot:qkd-pm} with these arguments up to (and including) \cref{step_pm:raw_key_gen}.\footnote{Note that for this, the other arguments in the description of \cref{prot:qkd-pm} are not used, so we do not need to specify them. The collective attack bound only depends on the protocol arguments specified in \cref{def:coll_attack}.}
For a collective attack $\cA: Q \to QE$, denote the state at the end of \cref{step_pm:raw_key_gen} as $\nu_{UVSIE}$.
Let $\nu_{UVSIEC}$ be an extension of this state, where $C = \ev(V, I, S)$.
A collective attack bound (for the choice of parameters fixed above) is a map $\ca: \mbP(\cC) \to \R$ such that for any collective attack $\cA$, the state $\nu_{UVSIEC}$ (which depends on $\cA$) satisfies 
\begin{align}
H(S | I E C)_{\nu} \geq \ca(\nu_{C}) \,. \label{eqn:def_ca_ineq}
\end{align}
\end{definition}

\subsection{Security against general attacks} \label{sec:gen_attacks}

Having introduced our framework for general prepare-and-measure protocols and collective attack bounds, we can now state the main technical result of this paper, namely that a collective attack bound implies a security statement against general attacks.
For this, we first recall the security definition for QKD, namely the notions of correctness, secrecy, and complete\-ness~\cite{renner_thesis}.
This security definition is \textit{composable}, meaning that the key generated by a protocol satisfying this definition can safely be used for other protocols~\cite{portmann2021security}.

\begin{definition}[Correctness, secrecy, and completeness] \label{def:qkd_sec}
Consider a QKD protocol in which Alice and Bob can decide whether or not to abort the protocol.
Let $\rho_{K \hat K E}$ be the final state at the end of the protocol (for a given initial state), where $K$ and $\hat K$ are Alice's and Bob's version of the final key, respectively, and $E$ contains all side information available to the adversary Eve at the end of the protocol.
The protocol is called $\eps^{\setft{cor}}$-correct, $\eps^{\setft{sec}}$-secret, and $\eps^{\setft{comp}}$-complete if the following holds: 
\begin{enumerate}
\item \emph{Correctness.} For any actions of the adversary Eve:
\begin{align*}
\pr{K \neq \hat K \wedge \setft{not abort}} \leq \eps^{\setft{cor}} \,.
\end{align*}
\item \emph{Secrecy.} For any actions of the adversary Eve:\footnote{Note that here and throughout the paper, we use the difference in trace norm, not the trace distance. The latter has an additional normalisation factor of $\frac{1}{2}$.}
\begin{align*}
\norm{\rho_{KE \wedge \Omega} - \tau_{K} \ot \rho_{E \wedge \Omega}}_1 \leq \eps^{\setft{sec}} \,,
\end{align*}
where $\tau_K$ is the maximally mixed state on system $K$, $\Omega$ is the event that the protocol does not abort, and $\rho_{\wedge \Omega} = \pr{\Omega} \rho_{|\Omega}$ is the subnormalised state conditioned on $\Omega$ (see \cref{sec:notation} for details).
\item \emph{Completeness.} For a given noise model for the protocol there exists an honest behaviour for the adversary Eve such that 
\begin{align*}
\pr{\setft{abort}} \leq \eps^{\setft{comp}} \,.
\end{align*}
\end{enumerate}
\end{definition}
Note that correctness and secrecy must hold for any behaviour of Eve (and also any noise model), while completeness is concerned with the honest implementation of the protocol.
Correctness and secrecy bound the probability of Alice and Bob receiving different or insecure keys without detecting this fact and aborting the protocol.
Completeness says that the protocol is robust against a given noise model in the sense that for this noise model, the probability of aborting the protocol is small if Eve behaves honestly.
It is common to combine the correctness and secrecy parameters and call a protocol $(\eps^{\setft{cor}} + \eps^{\setft{sec}}/2)$-secure, where the factor of $1/2$ arises because our definition of secrecy uses the difference in trace norm, not the trace distance, which has an additional factor of $1/2$.

Our main result is that \cref{prot:qkd-pm} satisfies the correctness and secrecy conditions. Formally, we show the following.
\begin{theorem} \label{thm:general_qkd_pm}
Fix any choice of arguments $n$, $\psi_{UQ}$, $\{N^{(v)}\}_{v \in \cV}$, $\pd$, $\rk$, $\ev$, $k_\ca$, $\lambda_{\ec}$, $\eps_{\kv}$, and $\eps_{\pa}$ for \cref{prot:qkd-pm}.
Let $\ca: \mbP(\cC) \to \R$ be an affine collective attack bound for this choice of arguments.
For any $\eps_s, \eps_a > 0$ and $\alpha \in (1, 3/2)$, choose a final key length $l$ that satisfies
\begin{multline}
l \leq n \, k_{\ca} - n \, \frac{\alpha-1}{2-\alpha} \, \frac{\ln(2)}{2} V^2 - \frac{g({\eps_s}) + \alpha \log(1/\eps_a)}{\alpha-1} \\
-  n \, \left( \frac{\alpha-1}{2-\alpha} \right)^2 K'(\alpha) 
- \lceil 2 \log(1/\eps_{\pa})\rceil - \lceil \log(1/\eps_{\kv}) \rceil - \lambda_{\ec} \,, \label{eqn:key_length}
\end{multline}
where $g({\eps_s})$, $V$, and $K'(\alpha)$ are defined in \cref{thm:with_testing}.
With this choice of parameters and assuming that \cref{cond:eve_seq} holds, \cref{prot:qkd-pm} is $\eps^{\setft{cor}}$-correct and $\eps^{\setft{sec}}$-secret for 
\begin{align*}
\eps^{\setft{cor}} = \eps_{\kv} \,, \qquad \eps^{\setft{sec}} = \max\{\eps_{\pa} + 4 \, \eps_s, 2\,\eps_a\} + 2 \, \eps_{\kv}\,.
\end{align*}
\end{theorem}
We prove this theorem in \cref{sec:main_proof}.
In addition, we also show completeness; since this is much more straightforward and only uses standard techniques, we defer this to \cref{app:completeness}.

\subsection{Sample application: B92 protocol}\label{sec:b92}
We now demonstrate how to apply our framework, using the B92 protocol as an example.
The B92 protocol has no natural entanglement-based analogue\footnote{By this, we mean an equivalent entanglement-based protocol that does not require ``artificial'' constraints on the reduced state on Alice's side and still achieves the same key rate as the prepare-and-measure version of B92. We note that some works~\cite{tamaki2003unconditionally,tamaki2004unconditional} do use an entanglement-based version of B92, but this causes a significant loss in asymptotic key rate.} and therefore cannot be analysed with the original EAT.
Nonetheless, it is very simple, and therefore provides arguably the easiest example to demonstrate the application of our framework to a protocol that cannot be analysed with the EAT.
Furthermore, while there exist analytic security proofs of B92 using entropic uncertainty relations~\cite{tamaki2003unconditionally,tamaki2004unconditional}, these techniques yield key rates that are far from optimal even in the asymptotic regime.
This is in contrast to highly symmetric protocols such as BB84, where entropic uncertainty relations yield essentially tight proofs~\citep{tomamichel2012tight}.


We emphasise that the purpose of this section is to illustrate our general results with a simple example, not to derive the tightest possible key rates for a particular protocols.
We leave the analysis of more complicated protocols, where deriving the collective attack bound may be more involved, for future work.
In \cref{app:decoy}, we also sketch how to express the decoy state BB84 protocol as an instance of our framework and how to derive a collective bound for it, demonstrating that the widely-used decoy state technique also naturally fits within our framework.

We also note that very recent work~\cite{george2022finite} has analysed the performance of the EAT on entanglement-based QKD protocols (and prepare-and-measure protocols that have a natural entanglement-based analogue) and found that it provides better key rates than previous methods.
Since our GEAT-based security proof produces essentially the same key rates as the EAT in cases where both methods can be applied, this suggests that our framework will provide very good key rates also in cases where the EAT cannot be applied.

We start by giving an informal description of the B92 protocol and the intuition behind it.
Then, we show how to view the B92 protocol as an instance of our general \cref{prot:qkd-pm}.
Using the technique from~\cref{sec:col_attack_bounds} to derive a collective attack bound, we can then apply~\cref{thm:general_qkd_pm} to obtain a security statement for general attacks.
To illustrate the result, we numerically compute the key rate for different choices of the number of rounds and tolerated noise level in \cref{sec:b92_key_rate}.

Each round of the B92 protocol works as follows: Alice chooses a bit $u \in \bits$ uniformly at random.
If $u = 0$, she prepares the state $\ket{\psi}_Q = \ket{0}$, whereas if $u = 1$, she prepares $\ket{\psi}_Q = \ket{+}$.\footnote{More generally, instead of $\ket{0}$ and $\ket{+}$, any two non-orthogonal states can be used. It has been observed that using states that are at a different angle to each other than $\ket{0}$ and $\ket{+}$ can be advantageous~\cite{coles2016numerical}. Since our goal is to provide an illustration, not optimise the key rate, we pick $\ket{0}$ and $\ket{+}$ for simplicity.}
She sends $\ket{\psi}_Q$ to Bob, who chooses $y \in \bits$ uniformly at random and measures the system $Q$ in the computational basis if $y = 0$ and the Hadamard basis if $y = 1$.
If he obtains outcome ``1'' (when measuring in the computational basis) or ``-'' (when measuring in the Hadamard basis), he sets $v = y \oplus 1$.
Otherwise, he sets $v = \bot$.
In the sifting step, Bob announces in which rounds he recorded $v = \bot$, and Alice sets $u = \bot$ for those rounds, too.
The bits $u$ and $v$ from all of the rounds form  the raw key.
To detect possible tampering by Eve, Alice and Bob compare their values of $u$ and $v$ on a subset of rounds.

The intuition behind this protocol is the following: the secret information that will make up the key is encoded in Alice's basis choice $u$ (where $u = 0$ corresponds to the computational and $u=1$ to the Hadamard basis).
When Bob receives the system $Q$ he tries to find out which basis the state was prepared in.
For this, he guesses a basis $y$ and measures $Q$ in this basis.
Suppose he chose $y = 0$, i.e.~the computational basis, and assume that Eve did not tamper with the system $Q$.
Then, if he obtains outcome ``1'' he concludes that Alice cannot have prepared the state $\ket{0}$ and therefore must have chosen $u = 1$.
Accordingly, he sets $v = 1 = y \oplus 1$.
If Bob obtains outcome ``0'' he cannot deduce Alice's basis choice as both the states $\ket{0}$ and $\ket{+}$ may produce outcome ``0'' when measured in the computational basis, so he sets $v = \bot$.
Likewise, if he chose $y = 1$ and obtains outcome ``-'', this  provides conclusive evidence that Alice cannot have prepared the state $\ket{+}$, so he sets $v = 0 = y \oplus 1$, whereas the outcome ``+'' is inconclusive.
If Eve tries to tamper with the system $Q$, she is likely to disturb the state as she does not know which basis it was prepared in.
Therefore, Alice and Bob will detect this tampering  when comparing their values of $u$ and $v$.

\subsubsection{B92 as an instance of {\cref{prot:qkd-pm}}} \label{sec:b92_formal}
We now give a more formal description of  the B92 protocol as an instance of \cref{prot:qkd-pm}.
As for the BB84 protocol described in \cref{sec:gen_pm_protocol}, this means specifying the arguments $\psi_{UQ}$, $\{N^{(v)}\}_{v \in\cV}$, $\pd$, $\rk$, and $\ev$.
For each round Alice chooses a bit $U_i$ uniformly at random and prepares $\ket{0}$ or $\ket{+}$ based on her choice, so
\begin{align*}
\psi_{UQ} = \frac{1}{2} ( \proj{0}_U \ot \proj{0}_Q + \proj{1}_U \ot \proj{+}_Q) \,.
\end{align*}
Bob measures in either the computational or Hadamard basis and uses the outcome to determine $V_i \in \{0, 1, \bot\}$ as described before. This measurement is described by the following POVM:
\begin{align*}
N^{(0)} = \frac{1}{2} \proj{-}\,, \; N^{(1)} = \frac{1}{2} \proj{1}, \; N^{(\bot)} = \frac{1}{2} (\proj{0} + \proj{+}) \,.
\end{align*}
During the public discussion phase, Bob informs Alice which rounds were inconclusive, i.e.~yielded outcome $\perp$.
Therefore, 
\begin{align*}
I_i = \pd(U_i, V_i) = \begin{cases}
\bot & \textnormal{if $V_i = \bot$,} \\
\top & \textnormal{otherwise.}
\end{cases}
\end{align*}
To generate her raw key $S^n$, Alice uses her bits $U_i$ and discards the rounds for which Bob's measurement outcome was inconclusive, which she knows from the value of $I_i$: 
\begin{align*}
S_i = \rk(U_i, I_i) = \begin{cases}
\bot & \textnormal{if $I_i = \bot$,} \\
U_i & \textnormal{otherwise.}
\end{cases}
\end{align*}
To generate the statistics $\hat C_i$, Bob will check whether his guess $\hat S^n$ for Alice's raw key agrees with his own raw data $V^n$.
As for the BB84 protocol described in \cref{sec:gen_pm_protocol}, Bob can only do so on a small fraction $\gamma$ of rounds because \cref{def:coll_attack} includes the classical statistics as a conditioning system.
Therefore, Bob chooses a value $T_i$ at random with $\pr{T_i = 1} = \gamma$ (the choice of $T_i$ can formally be included into $V_i$ or one can view $\ev$ as a randomised rather than deterministic function).
If $T_i = 0$, he sets $\hat C_i = \bot$, i.e.~$\ev_{T_i = 0}(V_i, I_i, \hat S_i) = \bot$.
Otherwise, he sets $\hat C_i = \ev_{T_i = 1}(V_i, I_i, \hat S_i)$ to
\begin{align}
 \begin{cases}
\texttt{fail} & \textnormal{if $\hat S_i = 0 \wedge V_i = 1$ or $\hat S_i = 1 \wedge V_i = 0$\,,} \\
\texttt{inc} & \textnormal{if $V_i = \bot$\,,}\\
\varnothing & \textnormal{else}.
\end{cases} \label{eqn:b92_c}
\end{align}
Of course, the functions $\ev_{T_i = 0}$ and $\ev_{T_i = 1}$ can be combined into  a single function $\ev$ to formally fit into the framework of \cref{prot:qkd-pm}.

\subsubsection{Collective attack bound} \label{sec:b92_col_attack}
We need to derive an affine collective bound $\ca(\nu_C) = \vec \lambda \cdot \vec \nu_C + c_{\vec \lambda}$ for the B92 protocol, where $\vec \nu_C$ denotes the probability vector of distribution $\nu_C$ as in \cref{sec:col_attack_bounds}.
For this, we use the steps and notation from \cref{sec:deriving_col_attack_bounds};
we recommend skipping this subsection on a first reading and returning to it after understanding \cref{sec:deriving_col_attack_bounds}.

In the notation of \cref{sec:deriving_col_attack_bounds}, the state $\tilde \psi_{PQ}$ is given by 
\begin{align*}
\tilde \psi_{PQ} = \frac{1}{\sqrt 2} ( \ket{0}_P \ot \ket{0}_Q + \ket{1}_P \ot \ket{+}_Q) \,.
\end{align*}
For any state $\hat \psi_{PQ}$ chosen by Eve, the statistics observed by Alice and Bob are described by 
\begin{align*}
\vec \nu_C = \tr{\vec \Gamma \hat \psi_{PQ}} \,,
\end{align*}
where $\vec \Gamma = (\Gamma_{\texttt{fail}}, \Gamma_{\texttt{inc}}, \Gamma_\varnothing, \Gamma_\bot)$ with
\begin{align*}
\Gamma_{\texttt{fail}} &= \gamma (\proj{0}_P \ot N^{(1)}_Q + \proj{1}_P \ot N^{(0)}_Q) \,,\\ 
\Gamma_{\texttt{inc}} &= \gamma \1_P \ot N^{(\bot)}_Q \,,\\
\Gamma_\varnothing &= \gamma (\1 - \Gamma_{\texttt{fail}} - \Gamma_{\texttt{inc}}) \,,\\
\Gamma_\bot &= (1 - \gamma) \1_P \ot \1_Q \,,
\end{align*}
and $\tr{\vec \Gamma \hat \psi_{PQ}}$ is shorthand for the vector of the traces with the individual elements of $\vec \Gamma$.
We can now directly apply the method from \cref{sec:deriving_col_attack_bounds} to find a collective attack bound $\ca(\nu_C) = \vec \lambda \cdot \vec \nu_C + c_{\vec \lambda}$: we can heuristically choose a $\vec \lambda$ and then determine $c_{\vec \lambda}$ by solving the convex optimisation problem from \cref{eqn:c_opt_psi} using the package Matlab \texttt{CVXQUAD}~\cite{fawzi2019semidefinite}.\footnote{One can pick $\vec \lambda$ by any numerical optimisation technique such as Matlab's $\texttt{fminsearch}$. Note that since $\vec \lambda$ can be chosen heuristically, it is not an issue if such an optimisation method does not have a convergence guarantee. In contrast, to determine $c_{\vec \lambda}$ one must use an optimisation method that guarantees a lower bound in order to ensure that the collective attack bound is valid. This is why it is important that $c_{\vec \lambda}$ be determined via a convex optimisation problem for which one can certify the solution by duality.}
For our numerical implementation, we employ additional simplifications to the optimisation problem from \cref{eqn:c_opt_psi} using the steps described in~\cref{app:b92_simpler}.
This helps with numerical performance, but is not strictly necessary.

\subsubsection{Key rate} \label{sec:b92_key_rate}

\begin{figure}[t!]
\centering
\begin{tikzpicture}
	\begin{axis}[
		height=6cm,
		width=8.3cm,
		xlabel=Depolarising probability $p$ in $\%$,
		ylabel=Key rate,
		xmin=0,
		xmax=0.095,
		ymax=0.25,
		ymin=0,
	     xtick={0,0.03,0.06,0.09,0.12,0.15},
	     xticklabels={0, 3, 6, 9, 12, 15},
          ytick={0,0.05,0.1,0.15,0.2,0.25},
          yticklabels={0,0.05,0.1,0.15,0.2,0.25},
          reverse legend,
		legend style={legend cell align=left,font=\footnotesize} 
	]
	
\input{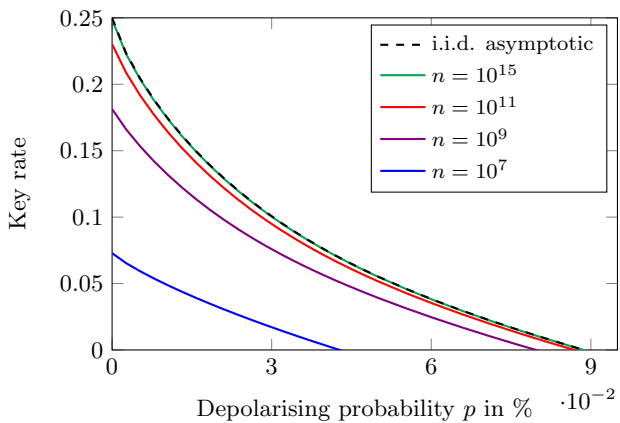}
\end{axis}
\end{tikzpicture}
\caption{Key rates for the B92 protocol as a function of the depolarising probability $p$ for $\eps^{\setft{cor}} = 5\cdot10^{-11}$, $\eps^{\setft{sec}} = 10^{-9}$, and $\eps^{\setft{comp}}=10^{-2}$.
The dashed line shows the key rate in the i.i.d.~asymptotic setting, i.e.~assuming that Eve behaves the same in each round and infinitely many rounds are executed.
We see that as the number $n$ of rounds in the protocol increases, the finite-size key rates against general attacks approach the i.i.d.~asymptotic rate.}
\label{fig:b92_curves}
\end{figure}

As our noise model for an honest implementation, we consider the depolarising channel with depolarising probability $p$, i.e.~the channel that maps $\rho \mapsto (1 - p) \rho + p \tau$, where $\tau$ is the maximally mixed state.
We determine the key rate as a function of $p$, i.e.~we determine the amount of key that can safely be generated from \emph{any potentially dishonest} implementation that produces the same statistics as the honest implementation with noise level $p$.
To this end, for every value of $p$ we first determine the statistics produced by an honest implementation with that noise level.
We then choose a collective attack bound and parameters for \cref{thm:general_qkd_pm} that ensure that the protocol is $\eps^{\setft{cor}}$-correct, $\eps^{\setft{sec}}$-secret, and $\eps^{\setft{comp}}$-complete for that noise level and $\eps^{\setft{cor}} = 5\cdot10^{-11}$, $\eps^{\setft{sec}} = 10^{-9}$, and $\eps^{\setft{comp}}=10^{-2}$.
Finally, we choose the key length to be the largest integer $l$ that satisfies the condition in \cref{eqn:key_length}.
We provide the choice of parameters in detail in \cref{sec:b92_params} and plot the resulting key rate in \cref{fig:b92_curves} for different numbers of rounds $n$.
We again note that the choice of parameters here is largely arbitrary and not optimised as the purpose of this example is only to illustrate the use of our general framework.

\section{Discussion}
We have introduced a proof technique for analysing the security of QKD protocols in the finite-size regime against general attacks.
This technique is best understood as a general procedure for converting a security proof in the i.i.d.~asymptotic setting into a finite-size security proof against general attacks.
To apply our technique, one can express a protocol of interest as an instance of our template \cref{prot:qkd-pm}, derive a collective attack bound (either using the general numerical technique described in \cref{sec:col_attack_bounds} or by reusing an existing analysis in the i.i.d.~asymptotic setting), and apply our \cref{thm:general_qkd_pm} to obtain finite-size key rates against general attacks.
Unlike previous techniques, our method can be applied directly to prepare-and-measure protocols and does not depend on the dimension of the underlying Hilbert space, allowing for a simple analysis of photonic prepare-and-measure protocols.

While we have provided a simple illustrative example of applying our framework to the well-known B92 protocol (\cref{sec:b92}), which is not amenable to treatment with the EAT, and sketched the analysis of the BB84 decoy-state protocol (\cref{app:decoy}), we leave it for future work to analyse more practical protocols and optimise the bounds one can obtain for those protocols.
This is especially relevant given that commercial QKD systems may become increasingly prevalent in the near future.
In particular, it would be interesting to see whether our framework can be used to prove the security of the differential phase-shift~\cite{dps_protocol} and coherent one-way~\cite{stucki2005fast} QKD protocols.
These protocols (and related ones using similar ideas) are relatively practical to implement, but notoriously hard to analyse.


\section{Methods}

\subsection{Notation} \label{sec:notation}
The set of states for a quantum system $A$ (with associated Hilbert space $\H_A$) is given by $\states(A) = \{\rho \in \pos(A) \, | \, \tr{\rho} = 1\}$, where $\pos(A)$ is the set of positive operators on $\H_A$. 
If $A$ is a quantum system and $X$ is a classical system with alphabet $\cX$, we call $\rho \in S(XA)$ a cq-state and can expand it as 
$\rho_{XA} = \sum_{x \in \cX} \proj{x} \ot \rho_{A, x}$
for subnormalised $\rho_{A, x} \in \pos(A)$. 
For $\Omega \subset \cX$, we define the partial and conditional states
\begin{align*}
\rho_{XA\wedge\Omega} = \sum_{x \in \Omega} \proj{x} \ot \rho_{A, x} \tand \rho_{XA|\Omega} = \frac{1}{\prs{\rho}{\Omega}} \rho_{XA\wedge\Omega} \,,
\end{align*}
where $\prs{\rho}{\Omega} \deq \tr{\rho_{XA\wedge\Omega}}$.
If $\Omega = \{x\}$, we also write $\rho_{XA|x}$ for $\rho_{XA|\Omega}$.
The set of quantum channels from system $A$ to $A'$ is denoted as $\cptp(A, A')$.
The trace norm (sum of the singular values) of an operator $L$ on $\cH_A$ is denoted as $\norm{L}_1$.

We will deal with two different entropies, the von Neumann entropy and the min-entropy, which are defined as follows.
Let $\rho_{AB} \in \states(AB)$ be a quantum state.
Then the conditional von Neumann entropy of  $A$ conditioned on $B$ is given by 
\begin{align*}
H(A|B)_\rho = - \tr{\rho_{AB} \log \rho_{AB}} + \tr{\rho_{B} \log \rho_{B}} \,.
\end{align*}
For $\eps \in [0,1]$, the $\eps$-smoothed min-entropy of $A$ conditioned on $B$ is
\begin{align*}
    H_{\min}^\eps(A|B)_{\rho} = - \log \inf_{\tilde \rho_{AB}} \inf_{\sigma_{B} \in \states(B)} \norm{\sigma_B^{-\frac{1}{2}} \tilde \rho_{AB} \sigma_B^{-\frac{1}{2}} }_{\infty} \, ,
\end{align*}
where $\norm{\cdot}_{\infty}$ denotes the spectral norm and the first infimum is taken over all states $\tilde \rho_{AB} \in \cB_{\eps}(\rho_{AB})$ in the $\eps$-ball around $\rho_{AB}$ (in terms of the purified distance~\cite{tomamichel2015quantum}).

\subsection{Universal hashing and randomness extraction}

To check that Alice's and Bob's keys are the same, our general QKD protocol will make use of a universal hash family, and to extract a secure key from Alice's and Bob's raw data we will use a randomness extractor.
Here, we briefly define what these primitives achieve.
We refer to~\cite{renner_thesis} for a more detailed exposition and explanation of their construction.

\begin{definition}[Universal hash family] \label{def:univ_hash}
Let $M$ be a set.
A family $\cF$ of functions from $M$ to $\bits^l$ with a probability distribution $P_{\cF}$ over $\cF$ is called a universal hash family if for any $x \neq x' \in M$, $\setft{Pr}_{f}[f(x) = f(x')] \leq 2^{-l}$.
\end{definition}

\begin{definition}[Quantum-proof strong extractor \cite{renner_thesis, konig2011sampling, de2012trevisan}] \label{def:extractor}
A function $\ext: \bits^m \times \bits^d \to \bits^l$ is a quantum-proof strong $(k, \eps_{\ext})$-extractor if for any $\rho_{SE} \in \pos(SE)$ with $\tr{\rho} \leq 1$ (and $S$ classical with dimension $2^m$) for which $\hmin(S|E)_{\rho} \geq k$, we have
\begin{align*}
\norm{\ext(\rho_{SE} \otimes \tau_D) - \tau_K \ot \rho_E \ot \tau_D}_1 \leq \eps_{\ext} \,,
\end{align*}
where $\tau_D$ and $\tau_K$ are maximally mixed states of dimension $2^d$ and $2^l$, respectively, and the map $\ext$ acts on the classical systems $S$ and $D$.
The input on system $D$ is called the \emph{seed} of the extractor.
\end{definition}

This definition of extractors makes use of the non-smoothed min-entropy $\hmin(S|E)_{\rho}$.
It is straightforward to modify this condition so that it only requires a lower bound on the smooth min-entropy:
if $\ext$ is a quantum-proof strong $(k, \eps_{\ext})$-extractor as in \cref{def:extractor} and $\rho_{SE}$ satisfies $\hmin^\eps(S|E)_{\rho} \geq k$, then
\begin{align}
\norm{\ext(\rho_{SE} \otimes \tau_D) - \tau_L \ot \rho_E \ot \tau_D}_1 \leq \eps_{\ext} + 4 \eps\,. \label{eqn:smoothed_extractor}
\end{align}
To see that this is the case, note that $\hmin^\eps(S|E)_{\rho} \geq k$ means  that there exists a $\rho'$ within $\eps$ purified distance of $\rho$ for which $\hmin(S|E)_{\rho'} \geq k$.
By the relation between purified distance and trace distance~\cite{tomamichel2015quantum}, we have $\norm{\rho - \rho'}_1 \leq 2 \eps$.
Then, \cref{eqn:smoothed_extractor} follows from the triangle inequality and because applying the map $\ext$ cannot increase the trace distance.

For the purposes of QKD, a simple construction based on two-universal hashing~\cite{renner_thesis} provides sufficiently good parameters.
We also note that more involved constructions exist that require shorter seeds, but this if typically not a concern for QKD applications (see e.g.~\cite{de2012trevisan} for a very efficient example using Trevisan's extractor).
\begin{lemma}[{\cite{renner_thesis}}]
There exist quantum-proof strong $(k, \eps_{\ext})$-extractors $\ext: \bits^m \times \bits^d \to \bits^l$ for $d = m$ and $l \leq k - 2 \log(1/\eps_{\ext})$.
\end{lemma}

\subsection{Generalised entropy accumulation} \label{sec:geat_prelim}

In this section, we introduce the GEAT from Ref.~\cite{gen_eat}.
Most of this section is taken directly from \cite{gen_eat} and we refer to the introduction of that paper for a more detailed description of the setting and how it compares to the EAT~\cite{eat}.
Consider a sequence of channels $\cM_i \in \cptp(R_{i-1} E_{i-1}, C_i A_i R_i E_i)$ for $i \in \{1, \dots, n\}$, where $C_i$ are classical systems with common alphabet $\cC$.
In the context of cryptographic protocols, one should think of $E_i$ as Eve's side information after the $i$-th round, $R_i$ as some internal system of a device, $A_i$ as the protocol's output in the $i$-th round, and $C_i$ as classical statistics that determine whether the protocol aborts (e.g.~by checking the number of rounds on which $A_i$ does not satisfy a certain property).
For all results in this paper, $R_i$ can be chosen to be trivial.
However, for (semi-)device-independent applications, the systems $R_i$ are important because they can be used to describe the internal memory of the untrusted devices.
As this is an interesting direction for future work, we state the theorem in full generality here.

We require that these channels $\cM_i$ satisfy the following condition: defining $\cM'_i  = \setft{Tr}_{C_i} \circ \cM_i$ (where $\setft{Tr}_{C_i}$ is the partial trace over system $C_i$ and $\circ$ is the composition of channels), there exists a channel $\cT \in \cptp(A^n E_n, C^n A^n E_n)$ such that $\cM_n \circ \dots \circ \cM_1 = \cT \circ \cM'_n \circ \dots \circ \cM'_1$ and $\cT$ has the form
\begin{multline}
\cT(\omega_{A^n E_n}) = \sum_{y \in \cY , z \in \cZ} (\Pi_{A^n}^{(y)} \otimes \Pi_{E_n}^{(z)}) \omega_{A^n E_n} (\Pi_{A^n}^{(y)} \otimes \Pi_{E_n}^{(z)}) \\ \otimes \proj{r(y,z)}_{C^n} \,, \label{eqn:measurement_condition}
\end{multline}
where $\{\Pi_{A^n}^{(y)}\}$ and $\{\Pi_{E_n}^{(z)}\}$ are families of mutually orthogonal projectors on $A_i$ and $E_i$, and  $r : \cY \times \cZ  \to \cC$ is a deterministic function.  
Intuitively, this condition says that the classical statistics can be reconstructed ``in a projective way'' from systems $A^n$ and $E_n$ at the end of the protocol.
In particular, this requirement is always satisfied if the statistics are computed from classical information contained in $A^n$ and $E_n$, which is the case for the applications in this paper.
We note that the statistics are still generated in a round-by-round manner; \cref{eqn:measurement_condition} merely asserts that they could be \textit{reconstructed} from the final state.

Let $\mbP$ be the set of probability distributions on the alphabet $\cC$ of $C_i$, and let $\tilde E_{i-1}$ be a system isomorphic to $R_{i-1} E_{i-1}$.
For any $q \in \mbP$ we define the set of states 
\begin{multline}
\label{eqn:def_sigma}
  \Sigma_i(q) = \bigl\{\nu_{C_i A_i R_i E_i \tilde E_{i-1}} = \cM_i(\omega_{R_{i-1} E_{i-1} \tilde E_{i-1}}) \,|\, \\ \omega \in \states(R_{i-1}E_{i-1}\tilde E_{i-1}) \text{ and } \nu_{C_i} = q \bigr\}  \ ,
\end{multline}
where $\nu_{C_i}$ denotes the probability distribution over $\cC$ with the probabilities given by $\pr{c} = \bra{c} \nu_{C_i} \ket{c}$. In other words, $\Sigma_i(q)$ is the set of states that can be produced at the output of the channel $\mathcal{M}_i$ and whose reduced state on $C_i$ is equal to the probability distribution $q$.
\begin{definition} \label{def:tradeoff}
A  function $f: \mbP \to \R$ is called a \emph{min-tradeoff function} for $\{\cM_i\}$ if it satisfies 
\begin{align*}
f(q) \leq \min_{\nu \in \Sigma_i(q)} H(A_i|E_i \tilde E_{i-1})_{\nu} \quad \forall i = 1, \dots, n\, .
\end{align*}
Note that if $\Sigma_i(q) = \emptyset$, then $f(q)$ can be chosen arbitrarily.
\end{definition}
Our result will depend on some simple properties of the tradeoff function, namely the maximum and minimum of $f$, the minimum of $f$ over valid distributions, and the maximum variance of $f$:
\begin{align*}
\Max{f} &\deq \max_{q \in \mathbb{P}} f(q) \,,\\
\Min{f} &\deq \min_{q \in \mathbb{P}} f(q) \,,\\
\MinSigma{f} &\deq \min_{q : \Sigma(q) \neq \emptyset} f(q) \,,\\
\Var{f} &\deq \max_{ q : \Sigma(q) \neq \emptyset} \sum_{x \in \cC} q(x) f(\delta_{x})^2 - \left(\sum_{x \in \cC} q(x) f(\delta_x) \right)^2 \,, \numberthis \label{eqn:def_var_f}
\end{align*}
where $\Sigma(q) = \bigcup_i \Sigma_i(q)$ and $\delta_x$ is the distribution with all the weight on element $x$.
We write $\freq{C^n}$ for the distribution on $\cC$ defined by $\freq{C^n}(c) = \frac{|\{i \in \{1,\dots,n\} : C_i = c\}|}{n}$.  
We also recall that in this context, an event $\Omega$ is defined by a subset of $\cC^n$, and for a state $\rho_{C^n A^n E_n R_n}$ we write $\prs{\rho}{\Omega} = \sum_{c^n \in \Omega}\tr{\rho_{A_1^n E_n R_n, c^n}}$ for the probability of the event $\Omega$ and
\begin{align*}
  \rho_{C^n A^n E_n R_n | \Omega} = \frac{1}{\prs{\rho}{\Omega}} \sum_{c^n \in \Omega} \proj{c^n}_{C^n} \otimes \rho_{A^n E_n R_n, c^n}
\end{align*}
for the state conditioned on~$\Omega$.
With this, we can finally state the GEAT of~\cite{gen_eat}.

\begin{theorem}[{GEAT~\cite{gen_eat}}]\label{thm:with_testing}
Consider a sequence of channels $\cM_i \in \cptp(R_{i-1}E_{i-1},\allowbreak C_i A_i R_i E_i)$ for $i \in \{1, \dots, n\}$, where $C_i$ are classical systems with common alphabet $\cC$ and the sequence $\{\cM_i\}$ satisfies \cref{eqn:measurement_condition} and the following \emph{no-signalling condition}:\footnote{In the special case relevant for prepare-and-measure QKD, where the systems $R_i$ are trivial, this no-signalling condition is trivially satisfied.} for each $\cM_i$, there exists a channel $\cR_i \in \cptp(E_{i-1}, E_i)$ such that $\setft{Tr}_{A_i R_i C_i} \circ \cM_i = \cR_i \circ \setft{Tr}_{R_{i-1}}$.
Let $\eps \in (0,1)$, $\alpha \in (1, 3/2)$, $\Omega \subset \cC^n$, $\rho_{R_0 E_0} \in \states(R_0 E_0)$, and $f$ be an affine min-tradeoff function with $h = \min_{c^n \in \Omega} f(\freq{c^n})$. Then,
\begin{multline}
\hmin^\eps(A^n | E_n)_{\cM_n \circ \dots \circ \cM_1(\rho_{R_0 E_0})_{|\Omega}} 
\geq n \, h - n \, \frac{\alpha-1}{2-\alpha} \, \frac{\ln(2)}{2} V^2 \\
- \frac{g(\eps) + \alpha \log(1/\prs{\rho^n}{\Omega})}{\alpha-1} -  n \, \left( \frac{\alpha-1}{2-\alpha} \right)^2 K'(\alpha)\, ,  \label{eqn:alpha_to_choose}
\end{multline}
where $\pr{\Omega}$ is the probability of observing event $\Omega$, and
\begin{align*}
g(\eps) &= - \log(1 - \sqrt{1-\eps^2}) \leq \log(2/\eps^2) \,,\\
V &= \log (2d_A^2+1) + \sqrt{2 + \Var{f}} \,,\\
K'(\alpha) &= \begin{multlined}[t]
    \frac{(2-\alpha)^3}{6 (3-2\,\alpha)^3 \ln 2} \, 2^{\frac{\alpha-1}{2-\alpha}(2\log d_{A}  + \Max{f} - \MinSigma{f})} \\ \ln^3\left( 2^{2\log d_{A} + \Max{f} - \MinSigma{f}} + e^2 \right) \,,
\end{multlined}
\end{align*}
with $d_A = \max_i \dim(A_i)$.
\end{theorem}

We briefly comment on the main differences between the GEAT as stated above and the EAT from~\cite{eat}.
The GEAT deals with a sequence of channels $\cM_i \in \cptp(R_{i-1}E_{i-1},\allowbreak C_i A_i R_i E_i)$ that can update both the internal memory register $R_i$ and the side information register $E_i$ (subject to the no-signalling condition), i.e.~change these states to e.g.~incorporate additional side information obtained in the protocol or account for measurements performed in response to the user's input.
In contrast, the EAT does not allow the side information register to be updated.
More formally, the EAT deals with channels $\cM_i' \in \cptp(R_{i-1},\allowbreak C_i A_i R_i I_i)$, where $I_i$ is side information produced in each round that cannot be updated in the future.
The final side information at the end of such a process is $E I^n$, where $E$ can be any additional side information from the initial state of the process that was never updated during the process.
If the side information registers $I_i$ satisfy the Markov condition $A^{i-1} \leftrightarrow  I^{i-1} E \leftrightarrow I_i$ (see~\cite{eat} for a more detailed explanation), then the EAT gives a lower bound on $\hmin^\eps(A^n | I^n E)_{\cM_n' \circ \dots \circ \cM_1'(\rho_{R_0})_{|\Omega}}$ similar to the one in \cref{thm:with_testing}.

We can now see at a high level why the EAT cannot be used to deal with prepare-and-measure protocols directly: in a prepare-and-measure protocol, the adversary Eve intercepts the quantum state sent from Alice to Bob in each round and updates her side information based on that.
Therefore, any technique used to deal with such protocols must allow for the side information to be updated like in the GEAT; the more restrictive scenario considered in the EAT does not capture this kind of protocol.

We also note that the GEAT is strictly more general than the EAT (see~\cite[Section 1]{gen_eat} for a proof).
Hence, any application that can be treated with the EAT can also be treated with the GEAT (up to some very minor loss in second-order parameters), and the resulting proofs are often much more straightforward; see~\cite[Section 5.2]{gen_eat} for an example.

\subsection{Proof of \cref{thm:general_qkd_pm}} \label{sec:main_proof}

In this section, we prove our main result, \cref{thm:general_qkd_pm}, i.e.~we show that \cref{prot:qkd-pm} is correct and secret.

\begin{proof}[Proof of \cref{thm:general_qkd_pm}]
For the correctness statement,  we need to show that $\pr{K \neq \hat K \wedge \setft{not abort}} \leq \eps_\kv$.
To see that this is the case, we note that due to the check in~\cref{step_pm:rkv}, the protocol not aborting implies that $\hash(S^n) = \hash(\hat S^n)$.
Furthermore, from~\cref{step_pm:pa} we see that $K \neq \hat K$ implies that $S^n \neq \hat S^n$.
Therefore, it suffices to show that 
\begin{align*}
\pr{S^n \neq \hat S^n \wedge \hash(S^n) = \hash(\hat S^n)} \leq \eps_{\kv} \,.
\end{align*}
Since Alice chooses the function $\hash$ at random from a universal hash family, this follows directly from~\cref{def:univ_hash} and completes the correctness proof.

The remainder of the proof will be concerned with the secrecy condition.
As explained in \cref{sec:eve_attack}, assuming~\cref{cond:eve_seq} we can model a general attack by a sequence of channels  
\begin{align*}
\cA_i: E'_{i-1} Q_i \to E'_i Q_i \,.
\end{align*}
Alice, Bob, and Eve's joint final state at the end of the protocol therefore contains systems 
\begin{align*}
U^nV^nI^nS^n\hat S^n \hat C^n K \hat K E'_n E' \,.
\end{align*}
Here, $E'_n$ is Eve's system after using the maps $\cA_1, \dots, \cA_n$, $E'$ stores the additional classical information published after~\cref{step_pm:ec}, i.e., the error correction information $\ec$, a description of the hash function $\hash$, the hash value $\hash(S^n)$, and the seed $\mu$, and the other systems are labelled as in \cref{prot:qkd-pm}.
This means that Eve's full side information is given by $I^n E'_n E'$.
Throughout the proof, we will denote the final state at the end of the protocol by $\rho_{U^nV^nI^nS^n\hat S^n C^n K \hat K E'_n E'}$.

By \cref{def:qkd_sec}, we need to show that
\begin{multline}
\norm{\rho_{K I^n E'_n E' \wedge \Omega} - \tau_{K} \ot \rho_{I^n E'_n E' \wedge \Omega}}_1 \\
\leq \max\{\eps_{\pa} + 4 \, \eps_s,  2\, \eps_a\} + 2 \, \eps_{\kv} \,, \label{eqn:secrecy_cond}
\end{multline}
where $\Omega$ is the event that the protocol does not abort and $\tau_{K}$ is the maximally mixed state on system $K$ of dimension $|K| = 2^l$.
Since the protocol's final state arises by application of a strong extractor in \cref{step_pm:pa}, we can reduce \cref{eqn:secrecy_cond} to an entropic statement.
This step requires careful technical treatment because the statistical check in~\cref{step_pm:stat} uses the systems $\hat C^n$, which are  computed from $\hat S^n$. However, $\hat S^n$ is Bob's guess for Alice's string $S^n$ and depends on the \emph{global} error correction information $\ec$, i.e., it cannot be generated in a round-by-round manner as required for the GEAT.
The intuition for circumventing this issue is as follows: 
if $\hat S^n \neq S^n$, then the protocol is likely to abort anyway because of~\cref{step_pm:rkv}; on the other hand, if $\hat S^n = S^n$, then we can replace $\hat S^n$ by $S^n$, and the latter is generated in a round-by-round manner.
Following this intuition, we can show that the entropy bound in \cref{claim:eat_qkd_pm} implies \cref{thm:general_qkd_pm}.
We give a formal proof of this step in \cref{app:entr_reduction} and continue here with proving the required entropy bound.
We also note that for protocols that include a separate parameter estimation step rather than using Bob's guess for Alice's raw key, \cref{claim:eat_qkd_pm} implies \cref{thm:general_qkd_pm} almost immediately.
\end{proof}

\begin{claim}\label{claim:eat_qkd_pm}
Let $\Omega_{C}$ be the event that $\ca(\freq{C^n}) \geq k_{\ca}$ (i.e.~the statistical check (\cref{step_pm:stat}) passes using the values $C^n$).
Continuing with the notation from before, for any $\alpha \in (1, 3/2)$:
\begin{multline}
\hmin^{\eps_s}(S^n |  I^n C^n E'_n)_{\rho_{| \Omega_C}} 
\geq n \, k_{\ca} - n \, \frac{\alpha-1}{2-\alpha} \, \frac{\ln(2)}{2} V^2 \\
- \frac{g({\eps_s}) + \alpha \log(1/\pr{\Omega_C})}{\alpha-1} -  n \, \left( \frac{\alpha-1}{2-\alpha} \right)^2 K'(\alpha) \,, \label{eqn:pm_entropy_bound}
\end{multline}
with $g({\eps_s})$, $V$, and $K'(\alpha)$ as in \cref{thm:with_testing}.
\end{claim}
\begin{proof}
To make use of the GEAT, we need to write $\rho_{S^n I^n C^n E'_n | \Omega_C}$ as the result of a sequential application of a quantum channel. For this we fix an attack $\cA_1, \dots, \cA_n$ and define
\begin{align*}
\cM_i : E'_{i-1} \to S_i I_i C_i E'_i
\end{align*}
as the following channel: given a quantum system $\omega_{E'_{i-1}}$, \begin{enumerate}
\item create the state $\psi_{U_i Q_i}$ (defined in \cref{step_pm:data_gen} of \cref{prot:qkd-pm}),
\item apply the attack map $\cA_i: Q_i E'_{i-1} \to Q_i E'_i$ to $\psi_{U_i Q_i} \ot \omega_{E'_{i-1}}$, \label{step_pm:attack}
\item measure $\{N^{(v)}\}_{v \in \cV}$ on system $Q_i$ and store the result in register $V_i$, \label{step_pm:meas}
\item set $I_i = \pd(U_i, V_i)$,
\item set $S_i = \rk(U_i, I_i)$,
\item set $C_i = \ev(V_i, I_i, S_i)$,
\item trace out registers $U_i$ and $V_i$ . \label{step_pm:traceout}
\end{enumerate}
Comparing the steps of the protocol and \cref{eqn:c_def} with this definition of $\cM_i$, we see that the marginal of $\rho$ on systems $S^n I^n C^n E'_n$ is the same as the output of the maps $\cM_i$: 
\begin{align*}
\rho_{S^n I^n C^n E'_n} = \cM_n \circ \dots \circ \cM_1(\omega_{E'_0}) \,,
\end{align*}
where $\omega_{E'_0}$ is the initial state of Eve's side information (which can be chosen to be trivial without loss of generality as explained in \cref{sec:eve_attack}).
If we define the systems $E_i = I^i C^i E'_i$, then by suitable tensoring with the identity map and copying the register $C_i$ we can view $\cM_i$ as a map 
\begin{align*}
\tilde \cM_i: E_{i-1} \to S_i E_i C_i \,.
\end{align*}
With this we can also express the final state (which technically now includes two copies of $C^n$, one explicit and one part of $E_n$) as
\begin{align*}
\rho_{S^n E_n C^n} = \tilde \cM_n \circ \dots \circ \tilde \cM_1(\omega_{E_0}) \,.
\end{align*}
With this notation, the entropy on the l.h.s.~of \cref{eqn:pm_entropy_bound} can be written as 
\begin{align*}
\hmin^{\eps_s}(S^n |  I^n C^n E'_n)_{\rho_{|\Omega_C}} = \hmin^{\eps_s}(S^n|E_n)_{\tilde \cM_n \circ \dots \circ \tilde \cM_1(\omega_{E_0})_{|\Omega_C}} \,.
\end{align*}
We want to apply \cref{thm:with_testing} to derive the desired lower bound in \cref{eqn:pm_entropy_bound}.
For this, we first need to check that the required conditions on the maps $\tilde \cM_i$ are satisfied.
The condition in \cref{eqn:measurement_condition} is clearly satisfied as the systems $C_i$ are themselves included in the conditioning system $E'_n$.
The non-signalling condition in \cref{thm:with_testing} is also trivially satisfied in this case since there is no system $R_i$.

We now need to argue that the collective attack bound $\ca: \mbP(\cC) \to \R$ used as an argument in \cref{prot:qkd-pm} is a min-tradeoff function for the maps $\{\tilde \cM_i\}$.
By \cref{def:tradeoff}, we need to show that for any $i$, attack $\cA_i: Q_i E'_{i-1} \to Q_i E'_i$ (in the definition of $\tilde \cM_i$, see \cref{step_pm:attack}), and state $\omega^{i-1}_{E_{i-1}\tilde E_{i-1}}$ (where $\tilde E_{i-1} \equiv E_{i-1}$), the following holds:
\begin{align}
\ca(\tilde \cM_i(\omega^{i-1})_{C_i}) \leq H(S_i | E_i \tilde E_{i-1})_{\tilde \cM_i(\omega^{i-1})} \,. \label{eqn:coll_att_mintradeoff}
\end{align}
For the rest of the proof, we fix an arbitrary choice of $i$, $\omega^{i-1}$, and $\cA_i$.
To relate \cref{eqn:coll_att_mintradeoff} to the definition of collective attack bounds (\cref{def:coll_attack}), we construct a collective attack $\cA': Q_i \to Q_i E_i \tilde E_{i-1}$ such that 
\begin{align}
\tilde \cM_i(\omega^{i-1})_{S_i C_i E_i \tilde E_{i-1}} = \nu_{S_i C_i E_i \tilde E_i} \,, \label{eqn:pm_nu_output}
\end{align}
where $\nu$ is defined as in~\cref{def:coll_attack}, i.e.~$\nu$ is the state produced by running a single round of~\cref{prot:qkd-pm} with the attack $\cA'$.\footnote{Of course, $\cA'$ will depend on $i$, $\omega^{i-1}$, and $\cA_i$. This is not a problem since \cref{def:coll_attack} holds for any collective attack, i.e., to show that \cref{eqn:coll_att_mintradeoff} holds for any $i$, $\omega^{i-1}$, and $\cA_i$, we can first fix an arbitrary choice, construct a ``custom'' collective attack that shows \cref{eqn:coll_att_mintradeoff} for that choice, and then apply the condition in \cref{def:coll_attack} to that choice.}
It is easy to check that \cref{eqn:pm_nu_output} is satisfied for the following choice of $\cA'$: given a state $\sigma_{Q}$, $\cA'$ first creates the (fixed) state $\omega^{i-1}_{E_{i-1}\tilde E_{i-1}}$ and then applies the (fixed) attack $\cA_i$ to $\sigma_{Q} \ot \omega^{i-1}_{E_{i-1}\tilde E_{i-1}}$ (with $Q_i = Q$).

Then, since $\ca$ is a collective attack bound, \cref{eqn:coll_att_mintradeoff} follows from \cref{def:coll_attack}: 
\begin{align*}
\ca(\tilde \cM_i(\omega^{i-1})_{C_i}) = \ca(\nu_{C_i})
&\leq H(S_i | E_i \tilde E_{i-1} C_i)_{\nu} \\
&= H(S_i | E_i \tilde E_{i-1})_{\tilde \cM_i(\omega^{i-1})} \,.
\end{align*}
Compared to~\cref{def:coll_attack}, we have dropped the explicit conditioning on $I \deq I_i$ since $I_i$ is already part of $E_i$, and in the last equality we can drop $C_i$ since it is also part of $E_i$.

This means that the function $\ca$ is a min-tradeoff function for \cref{prot:qkd-pm}.
By definition, for any $c^n \in \Omega_{C}$, $\ca(\freq{c^n}) \geq k_{\ca}$
Hence, \cref{claim:eat_qkd_pm} follows by applying \cref{thm:with_testing}.
\end{proof}

Having proved correctness and secrecy, we turn our attention to the completeness of \cref{prot:qkd-pm}, i.e.~we need to bound the probability that the protocol aborts when Eve does not interfere in the protocol, but the channel between Alice and Bob may be noisy.
In the protocol, Alice sends a quantum system $Q$ to Bob.
If the channel connecting Alice and Bob is noisy, instead of Alice's and Bob's joint state in each round being $\psi_{UQ}$, the joint state is $\cN(\psi_{UQ})$ for some channel $\cN: Q \to Q$.
This channel $\cN$ describes the noise model for \cref{prot:qkd-pm}.\footnote{Note that the channel $\cN$ is not something that needs to be added explicitly to the description of \cref{prot:qkd-pm}: formally, $\cN$ can be viewed as Eve's attack, i.e.~we can model the implementation of \cref{prot:qkd-pm} with a noisy channel and honest Eve by saying that Eve's attack is described by $\cN$. This also means that when we proved correctness and secrecy, we only needed to prove this for any behaviour of Eve, not any noise model, since the noise model can be included in Eve's actions.}

For a given noise model $\cN$, we need to choose the length of the error correction string  $\lambda_\ec$ to be sufficiently long such that Bob's guess $\hat S^n$ for Alice's raw key $S^n$ is correct with high probability, and as a consequence the check in \cref{step_pm:rkv} passes.
Furthermore, we need to choose the threshold $k_{\ca}$ to be sufficiently low that an honest noisy state passes \cref{step_pm:stat} with high probability.
The precise choice of parameters can be worked out using the properties of the error correcting code in \cref{step_pm:ec} and statistical tail bounds for \cref{step_pm:stat}.
We provide the details in \cref{app:completeness}.

\subsection{Deriving collective attack bounds} \label{sec:deriving_col_attack_bounds}
Our main result, \cref{thm:general_qkd_pm}, turns an \emph{affine} collective attack bound (defined in \cref{def:coll_attack}) into a security statement against general attacks.
Therefore, the main step one has to perform to use our framework is finding such an affine collective attack bound for a protocol of interest.
In this section, we give a numerical method for finding collective attack bounds for \cref{prot:qkd-pm} based on ideas from~\cite{coles2012unification,Winick2018reliablenumerical}.
Combined with \cref{thm:general_qkd_pm}, this means that the problem of finding key rate bounds against general attacks for any instance of \cref{prot:qkd-pm} is reduced to a numerical computation.

We begin by noting that we can rewrite the condition \cref{eqn:def_ca_ineq} from \cref{def:coll_attack} as follows: for any probability distribution $\nu^*_C \in \mbP(\cC)$ we require that
\begin{align}
\inf_{\nu \sth \nu_C = \nu^*_C} H(S | I E C)_{\nu} \geq \ca(\nu^*_{C}) \,, \label{eqn:entr_opt}
\end{align}
where the infimum is over all states $\nu$ that can result from a collective attack and have statistics $\nu^*_C$ (and the infinimum is infinite if there is no such state).
In the language of the GEAT, a collective attack bound essentially is a min-tradeoff function for a certain sequence of maps associated with \cref{prot:qkd-pm}.
More details on how a collective attack bound serves as a min-tradeoff function can be found in the proof of \cref{claim:eat_qkd_pm}

Since we are interested in an affine lower bound, we write the probability distribution $\nu_C$ as a probability vector $\vec \nu_C$ and, following~\cite{tan2021computing, tan2020improved}, make the ansatz 
\begin{align*}
\ca(\vec \nu_C) = \vec \lambda \cdot \vec \nu_C + c_{\vec \lambda}
\end{align*}
for some vector $\vec \lambda$ of the same dimension as $\vec \nu_C$ and a constant $c_{\vec \lambda}$.
We treat $\vec \lambda$ as a parameter that will be chosen heuristically. For example, one can choose $\vec \lambda$ by numerically estimating the gradient of the function $\nu'_C \mapsto \inf_{\nu \sth \nu_C = \nu'_C} H(S | I E C)_{\nu}$ around a particular choice of classical statistics $\nu_C^*$ that has been observed in an experimental realisation of the protocol, although this choice is not necessarily optimal and $\vec \lambda$ should be numerically optimised if one wants to obtain the best possible key rates.

Having chosen $\vec \lambda$ heuristically, we need to compute a value of $c_{\vec \lambda}$ that ensures that $\vec \lambda \cdot \vec \nu_C + c_{\vec \lambda}$ is a valid min-tradeoff function.
Inserting our ansatz into \cref{eqn:def_ca_ineq}, we see that for any fixed $\vec \lambda$, a valid choice of $c_{\vec \lambda}$ is one that satisfies
\begin{align}
c_{\vec \lambda} \leq \inf_{\nu} H(S | I E C)_{\nu} - \vec \lambda \cdot \vec \nu_C \,. \label{eqn:c_opt_general}
\end{align}
The infimum here is taken over the states $\nu$ described in \cref{def:coll_attack}.\footnote{To avoid confusion, we emphasise that the infimum here is taken over \emph{all} such states $\nu$, not just ones with a specific classical distribution $\nu^*_C$ as considered in \cref{eqn:entr_opt}. As explained in~\cite{tan2021computing}, one can view the optimisation in~\cref{eqn:c_opt_general} as arising from the Lagrange dual of \cref{eqn:entr_opt}, but we will not make use of this relation here explicitly.}

To tackle this optimisation problem, we consider an entanglement-based version of~\cref{prot:qkd-pm} using the source-replacement scheme explained in~\cite{coles2016numerical}.
As explained in \cref{sec:intro}, switching to an entanglement-based version of a prepare-and-measure protocol generally requires introducing ``artificial'' constraints on Eve's actions.
These artificial constraints are troublesome when applying the EAT to the entanglement-based version, but here we take a different approach: we only use the entanglement-based version to derive a collective attack bound (for which the artificial constraints do not present a problem).
This collective attack bound also applies to the original prepare-and-measure protocol and in \cref{thm:general_qkd_pm} we apply the EAT with this collective attack bound to the prepare-and-measure protocol directly.
We emphasise that the method for deriving a collective attack bound and our \cref{thm:general_qkd_pm} are entirely independent: \cref{thm:general_qkd_pm} does not depend on how the collective attack bound was derived and does not make use of an entanglement-based protocol itself.

In \cref{prot:qkd-pm} Alice prepares the state  
\begin{align*}
\psi_{UQ} = \sum_{u} p(u) \proj{u} \ot \proj{\psi}_{Q|u}
\end{align*}
and sends system $Q$ to Bob.
It is clear that Alice could equivalently prepare the state 
\begin{align*}
\ket{\tilde \psi}_{UQ} = \sum_{u} \sqrt{p(u)} \ket{u}_P \ot \ket{\psi}_{Q|u} \,,
\end{align*}
send system $Q$ to Bob, and only afterwards measure her own system $P$ in the computational basis, storing the outcome in system $U$.
Eve would now apply her collective attack $\cA: Q \to QE$ to system $Q$ of $\tilde \psi$, so the state after Eve's attack would be $\tilde \psi_{PQE}$.
We can replace this attack by giving Eve the ability to prepare a state $\hat \psi_{PQE}$ directly and distribute $P$ and $Q$ to Alice and Bob, respectively.
This kind of attack clearly gives Eve more power.
In fact, it gives Eve too much power: in order to still obtain a good key rate, we need to enforce the additional constraint that Alice's marginal of the state $\hat \psi$ is the same as her marginal of the state $\tilde \psi$ she would have prepared herself, i.e.~$\hat \psi_P = \tilde \psi_P$.
It is easy to see that even with this additional constraint, this latter kind of attack is still at least as general as any collective attack on the prepare-and-measure protocol described before.
Note that the condition $\hat \psi_A = \tilde \psi_A$ is not a physical constraint that Alice checks in an actual protocol, but rather the aforementioned additional artificial constraint.
Nonetheless, we can impose this artificial constraint on the optimisation problem used to calculate the collective attack bound.

For a fixed instance of \cref{prot:qkd-pm}, we can now view the state $\nu$ in \cref{def:coll_attack} as a function of $\hat \psi_{PQE}$:
\begin{multline*}
\nu_{ESIC}(\hat \psi) = \sum_{u, v} \ptr{PQ}{\proj{u}_P \ot N^{(v)}_Q \hat \psi_{PQE}} \ot \\ \projs{\rk(u, i)}_S \ot \projs{\pd(u, v)}_I \ot \projs{\ev(v, i)}_C \,.
\end{multline*}
Here, $\projs{\rk(u, \pd(u, v))}$ is shorthand for the projector $\proj{\rk(u, \pd(u, v))}$ and $i$ is shorthand for $\pd(u, v)$.
We can therefore write the optimisation problem from \cref{eqn:c_opt_general} as 
\begin{align*}
&\inf_{\hat \psi_{PQE}} H(S | I E C)_{\nu} - \vec \lambda \cdot \vec \nu_C \\
& \sth \hat \psi_{PQE} \geq 0\,, \quad \tr{\hat \psi_{PQE}} = 1\,,\quad \hat \psi_P = \tilde \psi_P\,,
\end{align*}
where $\nu = \nu(\hat \psi)$, and without loss of generality we can restrict the optimisation to pure states on $PQE$ with $E \equiv PQ$.

A lot of work in QKD has been focused on numerical methods for this kind of optimisation problem (see e.g.~\cite{coles2016numerical,Winick2018reliablenumerical,bunandar2020numerical,george2021num,hu2021robust}).
The key difficulty is that we need a \emph{lower} bound on the \emph{infimum} of a \emph{concave} function $H(S | I E C)_{\nu(\hat \psi)}$.
Here we use a method from~\cite{coles2012unification,Winick2018reliablenumerical} to turn this optimisation problem into a convex one.
As a first step, we observe that in the definition of $\nu$ we can incorporate the classical functions $\rk$, $\pd$, and $\ev$ into Alice's and Bob's measurements by defining
\begin{align}
M^{(s,i,c)}_{PQ} = \sum_{u, v:\, \bigg\{\substack{\rk(u, i) = s, \\\pd(u, v) = i, \\ \ev(v, i, s) = c}} \proj{u}_P \ot N^{(v)}_Q \,. \label{eqn:def_M}
\end{align}
Then, we can write $\nu_{ESIC}$ as 
\begin{align*}
\nu = \sum_{s, i, c} \ptr{PQ}{M^{(s, i, c)}_{PQ} \hat \psi_{PQE}} \ot \proj{s}_S \ot \proj{i}_I \ot \proj{c}_C \,.
\end{align*}
Remembering that we can assume that $\psi_{PQE}$ is pure, we now define the pure state
\begin{align*}
\ket{\nu^1} &= \sum_{s, i, c} \sqrt{M^{(s, i, c)}_{PQ}} \ket{\hat \psi}_{PQE} \ket{s}_S \ket{i}_I \ket{i}_{I'} \ket{c}_C \ket{c}_{C'} \,. 
\end{align*}
We observe that
\begin{align*}
\nu_{EIC} = \nu^1_{EIC} \,.
\end{align*}
Following the proof of \cite[Theorem 1]{coles2012unification}, a direct calculation shows that
\begin{align*}
H(S | I E C)_{\nu} = \dnormal{\nu^1_{PQSIC}}{\cP_S(\nu^1_{PQSIC})}
\end{align*}
where $\cP_S$ is the pinching map $\cP_S(\nu^1) = \sum_{s\in \cS} \proj{s}_S \nu^1 \proj{s}_S$.
We can view $\nu^1_{PQSIC}$ as a linear function of $\hat \psi_{PQ}$: 
\begin{multline}
\nu^1_{PQSIC}(\hat \psi_{PQ}) = \sum_{s, s', i, c} \sqrt{M^{(s, i, c)}_{PQ}} \hat \psi_{PQ} \sqrt{M^{(s', i, c)}_{PQ}} \\ 
\ot \ket{s}\!\bra{s'}_S \ot \proj{i}_I \ot \proj{c}_C \,, \label{eqn:def_nu1}
\end{multline}
Furthermore, the relative entropy is jointly convex.
Therefore, for a given $\vec \lambda$, a valid choice for $c_{\vec \lambda}$ can be found by solving the following convex optimisation problem:
\begin{align*}
c_{\vec \lambda} =
&\inf_{\hat \psi_{PQ}} \dnormal{\nu^1_{PQSIC}}{\cP_S(\nu^1_{PQSIC})} - \vec \lambda \cdot \vec \nu_C \\
& \sth \hat \psi_{PQ} \geq 0\,, \quad \tr{\hat \psi_{PQ}} = 1\,,\quad\hat \psi_P = \tilde \psi_P\,, \numberthis \label{eqn:c_opt_psi}
\end{align*}
where $\nu^1_{PQSIC}$ and $\nu_C$ are linear functions of $\hat \psi_{PQ}$.
To solve this optimisation problem, we can use standard techniques from convex optimisation.
In particular, in~\cite{fawzi2018efficient,fawzi2019semidefinite, fawzi_rational} techniques have been developed to bound the relative entropy from below by a sequence of semidefinite programs (SDPs).
These SDPs can then be solved using standard SDP solvers, and the solution to the dual SDP provides a certified lower bound.
Alternatively, one can also turn any feasible choice of $\hat \psi_{PQ}$ (ideally close to the optimal attack) into a certified lower bound using the techniques from~\cite{coles2016numerical,Winick2018reliablenumerical}.

We note that many protocols have additional structure that allow the optimisation problem in \cref{eqn:c_opt_psi} to be simplified before tackling it numerically.
Additionally, if the map $\ev$ from \cref{prot:qkd-pm} has a particular structure that distinguishes between ``test rounds'', in which Alice and Bob use their measurement outcomes to check whether Eve tampered with the protocol, and ``data rounds'', in which Alice and Bob generate the raw data for their key, the derivation of a collective attack bound can be further simplified.
We refer to \cite[Section V.A]{dupuis2019entropy} for a detailed explanation of this method and to \cref{sec:b92_formal} for an example of its use in our context.

\textit{Data availability.}
No experimental data was collected as part of this work. 

\textit{Code availability.}
Code for reproducing \cref{fig:b92_curves} is available from the authors upon request.

\textit{Acknowledgements.}
We thank Rotem Arnon-Friedman, Omar Fawzi, Marcus Haberland, Christoph Pacher, Joseph M.~Renes, Martin Sandfuchs, and David Sutter for helpful discussions.
We are especially grateful to Ernest Tan for helpful explanations regarding numerical methods for computing collective attack bounds.
This work was supported by the National Centres of Competence in Research (NCCRs) QSIT (funded by the Swiss National Science Foundation under grant number 51NF40-185902) and SwissMAP, the Air Force Office of Scientific Research (AFOSR) via project No. FA9550-19-1-0202, the SNSF project No. 200021\_188541 and the QuantERA project eDICT.

\textit{Author contributions.}
TM developed the proofs and wrote the manuscript with input from RR. RR guided and supervised the project.

\textit{Competing interests statement.} The authors declare no competing interests.

\bibliography{main}

\newpage
\onecolumngrid

\begin{appendices}
\crefalias{section}{appsec}

\section*{Supplementary Notes}
\section{Reduction to entropic condition} \label{app:entr_reduction}
In this section, we provide the detailed proof that \cref{claim:eat_qkd_pm} implies \cref{thm:general_qkd_pm} using the same ideas as Ref.~\cite[Section~4.2]{tan2020improved}.
We continue with the same notation as in the main text.
As a first step, we add to $\rho$ additional systems $C^n$ defined by
\begin{align}
C_i = \ev(V_i, I_i, S_i) \,. \label{eqn:c_def}
\end{align}
This means that the system $C_i$ is generated the same way as $\hat C_i$, except that we use Alice's actual raw key $S_i$ instead of Bob's guess $\hat S_i$.
We now define the following events (formally defined as subsets of possible values of the classical systems $S^n, \hat S^n, \hat C^n, C^n$, and $E'$): 
\begin{enumerate}[label=~]
\item $\Omega_{g}$: $S^n = \hat S^n$ (i.e.~Bob's guess of Alice's raw key is correct).
\item $\Omega_{\kv}$: $\hash(S^n) = \hash(\hat S^n)$ (i.e.~the raw key validation step (\cref{step_pm:rkv}) passes).
\item $\Omega_{\hat C}$: $\ca(\freq{\hat C^n}) \geq k_{\ca}$ (i.e.~the statistical check (\cref{step_pm:stat}) passes using the values $\hat C^n$).
\item $\Omega_{C}$: $\ca(\freq{C^n}) \geq k_{\ca}$ (i.e.~the statistical check (\cref{step_pm:stat}) passes using the values $C^n$).
\end{enumerate}
The event $\Omega$ of the protocol not aborting is $\Omega = \Omega_\kv \wedge \Omega_{\hat C}$.
If $S^n = \hat S^n$, then $\hash(S^n) = \hash(\hat S^n)$ and $C^n = \hat C^n$.
Therefore, 
\begin{align*}
\Omega \wedge \Omega_g = \Omega_\kv \wedge \Omega_{\hat C} \wedge \Omega_g = \Omega_{C} \wedge \Omega_g \,.
\end{align*}
Since \cref{step_pm:rkv} employs a universal hash function, the probability that the protocol does not abort despite $S^n \neq \hat S^n$ is at most $\eps_{\kv}$, i.e.~$\pr{\Omega_g^c \wedge \Omega} \leq \eps_{\kv}$, where $\Omega_g^c$ is the complement of $\Omega_g$.
Hence, we can bound the l.h.s.~of \cref{eqn:secrecy_cond} by
\begin{align*}
& \norm{\rho_{K I^n E'_n E' \wedge \Omega} - \tau_{K} \ot \rho_{I^n E'_n E' \wedge \Omega}}_1 \\
&\leq \norm{\rho_{K I^n E'_n E' \wedge \Omega \wedge \Omega_g} - \tau_{K} \ot \rho_{I^n E'_n E' \wedge \Omega \wedge \Omega_g}}_1 + \norm{\rho_{K I^n E'_n E' \wedge \Omega \wedge \Omega_g^c}}_1 + \norm{\tau_{K} \ot \rho_{I^n E'_n E' \wedge \Omega \wedge \Omega_g^c}}_1 \\
&\leq \norm{\rho_{K I^n E'_n E' \wedge \Omega_C \wedge \Omega_g} - \tau_{K} \ot \rho_{I^n E'_n E' \wedge \Omega_C \wedge \Omega_g}}_1 + 2\,\eps_\kv \,. \numberthis \label{eqn:sec_bound1}
\end{align*}
For the remainder of the proof, we will assume that 
\begin{align}
\pr{\Omega_g | \Omega_C} \geq \eps_s \tand \pr{\Omega_C} \geq \eps_a \label{eqn:min_eps_cond} \,.
\end{align}
This assumption is justified by the fact that otherwise, we  have 
\begin{align*}
\pr{\Omega_C \wedge\ \Omega_g} \leq \min \{\pr{\Omega_g | \Omega_C} , \pr{\Omega_C}\} \leq \max\{\eps_s, \eps_a\} \,,
\end{align*}
in which case the theorem statement follows directly from \cref{eqn:sec_bound1} and  
\begin{align*}
\norm{\rho_{K I^n E'_n E' \wedge \Omega_C \wedge \Omega_g} - \tau_{K} \ot \rho_{I^n E'_n E' \wedge \Omega_C \wedge \Omega_g}}_1
&\leq \norm{\rho_{K I^n E'_n E' \wedge \Omega_C \wedge \Omega_g}}_1 + \norm{\tau_{K} \ot \rho_{I^n E'_n E' \wedge \Omega_C \wedge \Omega_g}}_1 \\
&= 2 \, \pr{\Omega_C \wedge\ \Omega_g} \,.
\end{align*}
Defining the state $\rho_{K I^n E'_n E' (| \Omega_C) (\wedge \Omega_g)} = 1/\pr{\Omega_C} \rho_{K I^n E'_n E' \wedge \Omega_C \wedge \Omega_g}$, clearly 
\begin{align*}
\norm{\rho_{K I^n E'_n E' \wedge \Omega_C \wedge \Omega_g} - \tau_{K} \ot \rho_{I^n E'_n E' \wedge \Omega_C \wedge \Omega_g}}_1 \leq \norm{\rho_{K I^n E'_n E' (| \Omega_C) (\wedge \Omega_g)} - \tau_{K} \ot \rho_{I^n E'_n E' (| \Omega_C) (\wedge \Omega_g)}}_1 \,.
\end{align*}
Therefore, to show the theorem, it suffices to show that 
\begin{align}
\norm{\rho_{K I^n E'_n E' (| \Omega_C) (\wedge \Omega_g)} - \tau_{K} \ot \rho_{I^n E'_n E' (| \Omega_C) (\wedge \Omega_g)}}_1 \leq \eps_{\pa} + 4 \, \eps_s \,. \label{eqn:secrecy_cond_mod}
\end{align}

The state $\rho_{K I^n E'_n E' (| \Omega_C) (\wedge \Omega_g)}$ is produced by applying a strong extractor in \cref{step_pm:pa}.
Comparing \cref{def:extractor} and \cref{eqn:secrecy_cond_mod} and remembering that the seed $\mu$ is chosen uniformly at random by Alice and is part of the system $E'$, we see that we need  to show that 
\begin{align*}
\hmin^{\eps_s}(S^n | I^n E'_n E')_{\rho_{K E'_n E' (| \Omega_C) (\wedge \Omega_g)}} \geq l + \lceil 2 \log(1/\eps_{\pa})\rceil \,.
\end{align*}
To this end, we can first bound the l.h.s.~by
\begin{align*}
\hmin^{\eps_s}(S^n | I^n E'_n E')_{\rho_{(| \Omega_C) (\wedge \Omega_g)}} 
&\geq \hmin^{\eps_s}(S^n | I^n E'_n E')_{\rho_{| \Omega_C}} \\
&\geq \hmin^{\eps_s}(S^n | I^n C^n E'_n E')_{\rho_{| \Omega_C}}\\
&\geq \hmin^{\eps_s}(S^n | I^n C^n E'_n)_{\rho_{| \Omega_C}} - \lambda_{\ec} - \lceil \log(1/\eps_{\kv}) \rceil \,,  \numberthis \label{eqn:sub_lambdaec}
\end{align*}
where the first inequality follows from \cite[Lemma 10]{tomamichel2017largely} together with the first condition in \cref{eqn:min_eps_cond}, the second inequality holds because conditioning on additional classical information $C^n$ can only decrease the min-entropy, and the third inequality is a chain rule for the min-entropy which uses the fact  that the only information in $E'$ that is correlated with $S^n$ is the error correction information $\ec \in \bits^{\lambda_{\ec}}$ and the hash $\hash(S^n) \in \bits^{\lceil \log(1/\eps_{\kv}) \rceil}$.
We can show a lower bound on $\hmin^{\eps_s}(S^n | I^n C^n E)_{\rho_{| \Omega_C}}$ using the GEAT.
Since this is the core of the security proof, we present it as a separate claim.
From the preceding discussion and the assumption $\pr{\Omega_C} \geq \eps_a$ it is then clear that \cref{claim:eat_qkd_pm} implies \cref{thm:general_qkd_pm}.

\section{Completeness of Protocol \ref{prot:qkd-pm}} \label{app:completeness}
To prove the completeness of \cref{prot:qkd-pm}, we need to bound the probability of the protocol aborting for a given choice of parameters and noise model.
Throughout this section, for a fixed choice of arguments in \cref{prot:qkd-pm} and a fixed noise model $\cN$, we denote by $\nu^{\setft{hon}}$ the corresponding ``honest single-round state'', i.e.~formally the state $\nu^{\setft{hon}}$ from \cref{def:coll_attack} when one chooses Eve's collective attack as $\cN$.
Furthermore, we assume that in \cref{step_ent:ec} Alice and Bob use the one-way error correction protocol from~\cite{error_corr}, which is essentially optimal.
We note that in \cref{prot:qkd-pm} the choice of error correction protocol has no effect on the security statement, only on the completeness statement, so one can also use a heuristic protocol with some fixed leakage $\lambda_{\ec}$ leading to a heuristic value of $\eps^{\setft{comp}}_\kv$ without impacting the security of the protocol.

There are two steps in which \cref{prot:qkd-pm} may abort:
\cref{step_pm:rkv} if $\hash(S^n) \neq \hash(\hat S^n)$, and \cref{step_pm:stat} if $\ca(\freq{\hat C^n}) < k_{\ca}$.
Since $S^n = \hat S^n$ implies $\hash(S^n) = \hash(\hat S^n)$, we can bound the total abort probability by $\pr{S^n \neq \hat S^n} + \pr{S^n = \hat S^n \wedge \ca(\freq{\hat C^n}) < k_{\ca}}$.
We denote these probabilities by $\eps^{\setft{comp}}_\kv$ and $\eps^{\setft{comp}}_\ev$, respectively, so the total abort probability is bounded by $\eps^{\setft{comp}}_\kv + \eps^{\setft{comp}}_\ev$.

\begin{lemma} \label{lem:comp_ec}
Fix a noise model $\cN$ and a choice of arguments in \cref{prot:qkd-pm}.
Then, any desired value of $\eps^{\setft{comp}}_\kv$ can be achieved as long as the following condition holds:
\begin{align*}
\lambda_{\ec} \geq n H(S | V I)_{\nu^{\setft{hon}}} + 2 \sqrt{n} \sqrt{1 - 2 \log (\eps^{\setft{comp}}_\kv / 2)} \log(1 + 2 |\cS|) + 2 \log \frac{2}{\eps^{\setft{comp}}_\kv} \,.
\end{align*}
\end{lemma}
\begin{proof}
By \cite{error_corr}, it suffices to show that 
\begin{align*}
\lambda_{\ec} \geq \hmax^{\tilde \eps}(S^n | V^n I^n)_{(\nu^{\setft{hon}})^{\ot n}} + 2 \log \frac{1}{\eps^{\setft{comp}}_\kv - \tilde \eps} \,,
\end{align*}
where $\tilde \eps \in [0, \eps^{\setft{comp}}_\kv)$ is a parameter that can be optimised over.
For simplicity, here we choose $\tilde \eps = \eps^{\setft{comp}}_\kv / 2$, but note that one could numerically optimise over $\tilde \eps$ if one wishes to derive the best possible completeness error.
By the $\hmax$-version of \cite[Corollary 4.10]{eat}, 
\begin{align*}
\hmax^{\eps^{\setft{comp}}_\kv / 2}(S^n | V^n I^n) \leq n H(S | V I)_{\nu^{\setft{hon}}} + 2 \sqrt{n} \sqrt{1 - 2 \log (\eps^{\setft{comp}}_\kv / 2)} \log(1 + 2 |\cS|) \,.
\end{align*}
Combining these two equations yields the lemma.
\end{proof}

\begin{lemma} \label{lem:comp_ev}
Fix a noise model $\cN$ and a choice of arguments in \cref{prot:qkd-pm}.
Then, any desired value of $\eps^{\setft{comp}}_\ev$ can be achieved as long as the following condition holds for $\delta = \ca(\nu^{\setft{hon}}_C) - k_{\ca}$: 
\begin{align*}
\delta^2 \geq \frac{2 (\Max{\ca} - \Min{\ca}) \delta + 6 \Var{\ca}}{3 n} \log \frac{1}{\eps^{\setft{comp}}_\ev} \,.
\end{align*}
\end{lemma}
\begin{proof}

Using the definitions at the start of \cref{app:entr_reduction}, we can write $\eps^{\setft{comp}}_\ev = \pr{\Omega_g \wedge \Omega_{\hat C}^c}$.
Since $\Omega_g \wedge \Omega_{\hat C}^c = \Omega_g \wedge \Omega_{C}^c$, we can bound $\eps^{\setft{comp}}_\ev \leq \pr{\Omega_C^c}$. 
The honest implementation is i.i.d., so the state we need to consider in \cref{step_pm:stat} is $(\nu^{\setft{hon}})^{\ot n}$.
Let $C_1, \dots, C_n$ be i.i.d.~random variables with distribution $\nu^{\setft{hon}}_C$.
Then we can view the value $k$ computed by Bob in \cref{step_pm:stat} as a random variable, too, and because $\ca$ is affine we have 
\begin{align*}
k = \frac{1}{n} \sum_{i = 1}^n \ca(\delta_{C_i}) \,,\qquad \ca(\nu^{\setft{hon}}_C) = \E \ca(\delta_{C_i}) \,,
\end{align*}
where by $\ca(\delta_{C_i})$ we mean the random variable that maps a value $c \in \cC$ of $C_i$ to $\ca(\delta_c)$, with $\delta_c$ the point distribution with all the weight on element $c$.
Conditioned on $S^n = \hat S^n$, \cref{step_pm:stat} aborts if $k < k_{\ca}$.
We can now apply Bernstein's inequality and find that 
\begin{align*}
\pr{\ca(\nu^{\setft{hon}}_C) - k > \delta} \leq \exp\left(- n \frac{\delta^2 / 2}{\setft{Var}[\ca(\delta_{C_i})] + (\Max{\ca} - \Min{\ca}) \delta / 3}\right) \,.
\end{align*}
Requiring that this be less than $\eps^{\setft{comp}}_\ev$ and noting that by \cref{eqn:def_var_f}, $\setft{Var}[\ca(\delta_{C_i})] \leq \Var{\ca}$, we find the desired result. 
\end{proof}
We note that for simple protocols, e.g.~ones where a certain number of rounds are used as ``test rounds'' that can either pass or fail, the above bound can usually be replaced by a simpler and tighter one using Hoeffding's inequality (see e.g.~\cite[Section 3.2]{arnon2019simple}).

\section{Relaxing the sequentiality assumption \label{app:block_seq}} 
\begin{figure*}[t]
\centering
\includegraphics[width=0.4\textwidth]{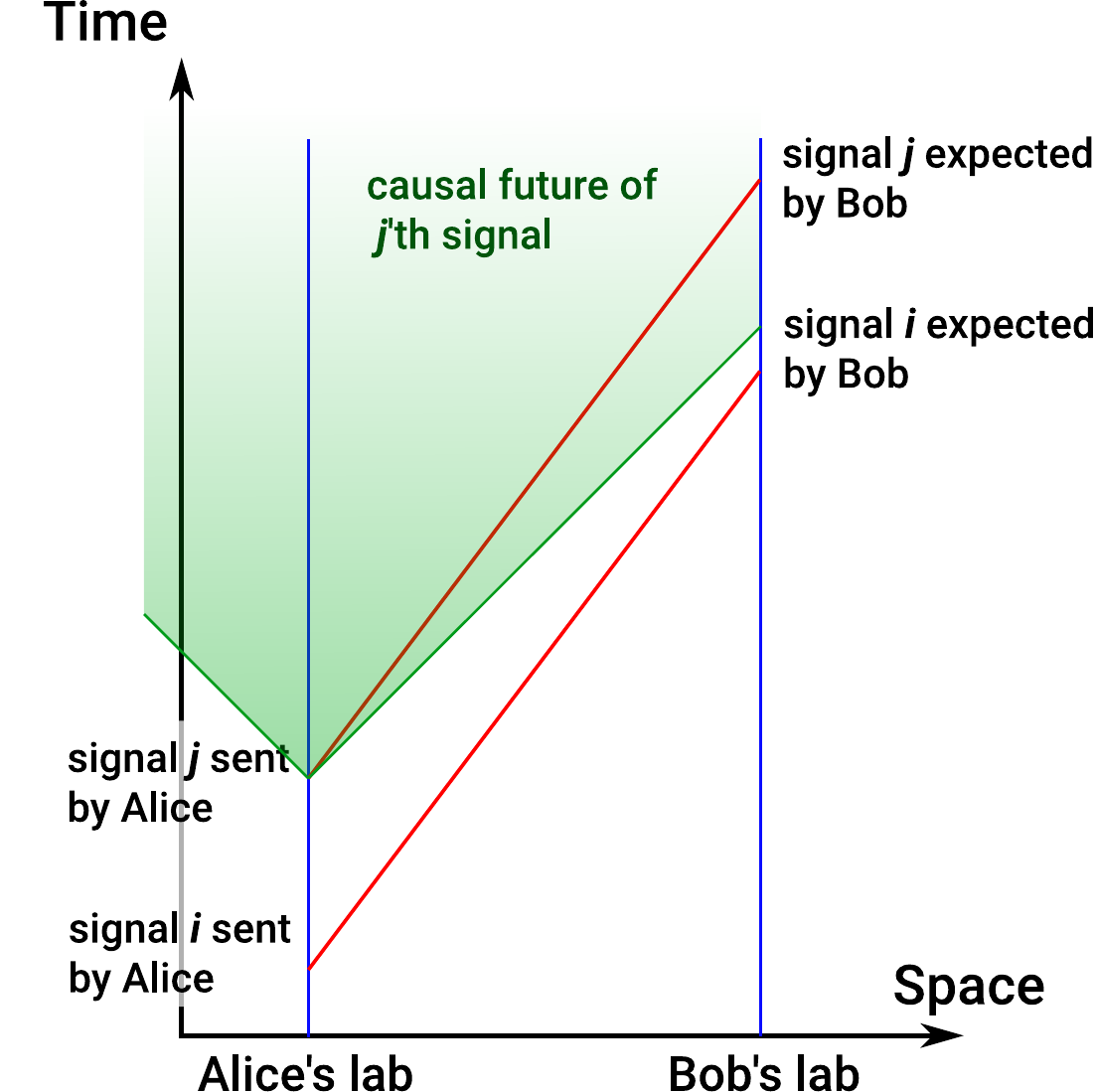}
\caption{Spacetime diagram illustrating causal structure of signal transmission.
If Alice and Bob send the signals on a pre-agreed schedule, the spacetime points where Alice sends signals and Bob expects signals are fixed.
The causal future (green shaded region) of signal $j$ contains all spacetime points to which Eve can transmit information about signal $j$.
If Eve were not able to speed up the signals at all, then the boundary of the causal future (green line) would coincide with Alice's and Bob's expected signal transmission speed (red line) and the sequentiality assumption would be ensured already between subsequent signals (i.e.~\cref{cond:eve_seq}).
If Eve can speed up the signals as shown in the figure, Alice and Bob can choose $i$ and $j$ sufficiently far apart (i.e.~choose a sufficiently large step size in \cref{cond:eve_seq_blocks}) that the expected arrival of signal $i$ by Bob lies outside the causal future of signal $j$ sent by Alice.}
\label{fig:spacetime}
\end{figure*}
As explained in \cref{sec:eve_attack}, for QKD implementations that allow the adversary Eve to speed up up the transmission of signals, it may be difficult to enforce \cref{cond:eve_seq} without significantly lowering the frequency with which signals are sent from Alice to Bob.
It is therefore useful to relax \cref{cond:eve_seq} and allow Eve to be in possession of $s$ signals at a time.
More formally, using the same notation as in \cref{sec:eve_attack}, we would like to prove security of a prepare-and-measure protocol under the following weaker condition:
\begin{condition} \label{cond:eve_seq_blocks}
Eve can only be in possession of at most $s$ subsequent systems $Q_i, \dots, Q_{i + s - 1}$ at the same time.
We call $s$ the \emph{step size} of Eve's attack.
\end{condition}
We note that this condition does \emph{not} mean that Eve has to process the signals in disjoint blocks of size $s$; instead, this condition allows Eve to e.g.~first apply an attack on systems $Q_1, \dots, Q_{s}$, then send system $Q_1$ to Bob and receive $Q_{s+1}$ from Alice, apply an attack to $Q_2, \dots, Q_{s+1}$, etc.
In other words, Eve can execute a ``rolling attack'' that always uses $s$ adjacent signals.
As explained in \cref{sec:eve_attack}, this condition can be enforced by Alice and Bob using a pre-agreed schedule on which to send their signals, assuming we can place some bound on the amount by which Eve could speed up the transmission of signals from Alice to Bob.
(The trivial bound is of course always that Eve can speed up the signal transmission to the speed of light.)
We emphasise that to enforce \cref{cond:eve_seq_blocks} (for some appropriately chosen $s$), Alice and Bob do \emph{not} need to lower the frequency with which they send signals or divide their signals into blocks with breaks between blocks.

We can now prove an analogous statement to our main result \cref{thm:general_qkd_pm} that only requires the weaker \cref{cond:eve_seq_blocks} instead of \cref{cond:eve_seq}.
The cost that we have to pay for allowing the weaker condition \cref{cond:eve_seq_blocks} is that the second-order term $\frac{g({\eps_s}) + \alpha \log(1/\eps_a)}{\alpha-1}$ from \cref{thm:general_qkd_pm} now acquires a prefactor $s$ (and, less importantly, $\eps_s$ gets replaced by $\eps_s/(3s-2)$) and we get an additional term $(s-1)g(\eps_s/(3s-2))$; the latter is negligible compared to the former because $\alpha$ is close to 1.
However, the first-order term remains unchanged and is independent of $s$, so in particular the asymptotic key rate (against general attacks with any fixed step size) is the same as in \cref{thm:general_qkd_pm}.
We illustrate this for the example of B92 in \cref{fig:b92_curves_blocked}.
\begin{theorem} \label{thm:block_qkd_pm}
Fix any choice of arguments $n$, $\psi_{UQ}$, $\{N^{(v)}\}_{v \in \cV}$, $\pd$, $\rk$, $\ev$, $k_\ca$, $\lambda_\ec$, $\eps_{\kv}$, and $\eps_{\pa}$ for \cref{prot:qkd-pm}.
Let $\ca: \mbP(\cC) \to \R$ be an affine collective attack bound for this choice of arguments.
For any $\eps_s, \eps_a > 0, \alpha \in (1, 3/2)$, and $s \in \N$, choose a final key length $l$ that satisfies
\begin{multline}
l \leq n \, k_{\ca} - n \, \frac{\alpha-1}{2-\alpha} \, \frac{\ln(2)}{2} V^2 - s \, \frac{g\left(\eps'\right) + \alpha \log(1/\eps_a)}{\alpha-1} \\
-  n \, \left( \frac{\alpha-1}{2-\alpha} \right)^2 K'(\alpha) - (s-1) g(\eps') - \lceil 2 \log(1/\eps_{\pa})\rceil - \lceil \log(1/\eps_{\kv}) \rceil - \lambda_{\ec} \,, \label{eqn:key_length_blocks}
\end{multline}
where $g(\cdot)$, $V$, and $K'(\cdot)$ are defined in \cref{thm:with_testing} and $\eps' \deq \frac{\eps_s}{3s-2}$.
With this choice of parameters and assuming that \cref{cond:eve_seq_blocks} holds for the value of $s$ chosen above, \cref{prot:qkd-pm} is $\eps^{\setft{cor}}$-correct and $\eps^{\setft{sec}}$-secret for 
\begin{align*}
\eps^{\setft{cor}} = \eps_{\kv} \,, \qquad \eps^{\setft{sec}} = \max\{\eps_{\pa} + 4 \, \eps_s, 2\,\eps_a\} + 2 \, \eps_{\kv}\,.
\end{align*}
\end{theorem}

The proof of \cref{thm:block_qkd_pm} follows the same steps as the proof of \cref{thm:general_qkd_pm} in \cref{sec:main_proof}, except that we will need to make some modifications to account for the more general structure of the attack.
Since most of the proof is identical, we only provide a sketch and point out the main differences compared to \cref{sec:main_proof}.

As in \cref{sec:main_proof}, we again denote the final state at the end of \cref{prot:qkd-pm} (for any fixed attack of the form above) by 
\begin{align*}
\rho_{U^nV^nI^nS^n\hat S^n C^n K \hat K E'_n E'} \,.
\end{align*}
The system labels here are as in \cref{sec:main_proof}.
The reduction from \cref{thm:general_qkd_pm} to \cref{claim:eat_qkd_pm} did not use the structure of the attack, so the same steps also allow us to reduce \cref{thm:block_qkd_pm} to the following claim.
\begin{claim}\label{claim:eat_block_pm}
Let $\Omega_{C}$ be the event that $\ca(\freq{C^n}) \geq k_{\ca}$ (i.e.~the statistical check (\cref{step_pm:stat}) passes using the values $C^n$).
Then, for any $\alpha \in (1, 3/2)$:
\begin{align}
\hmin^{\eps_s}(S^n |  I^n C^n E'_n)_{\rho_{| \Omega_C}} 
\geq n \, k_{\ca} - n \, \frac{\alpha-1}{2-\alpha} \, \frac{\ln(2)}{2} V^2 - s \, \frac{g\left(\eps'\right) + \alpha \log(1/\pr{\Omega_C})}{\alpha-1} -  n \, \left( \frac{\alpha-1}{2-\alpha} \right)^2 K'(\alpha) - (s-1) g(\eps')\,, \label{eqn:pm_entropy_bound_blocks}
\end{align}
with $g(\cdot)$, $V$, and $K'(\cdot)$ as in \cref{thm:with_testing} and $\eps' \deq \frac{\eps_s}{3s-2}$.
\end{claim}

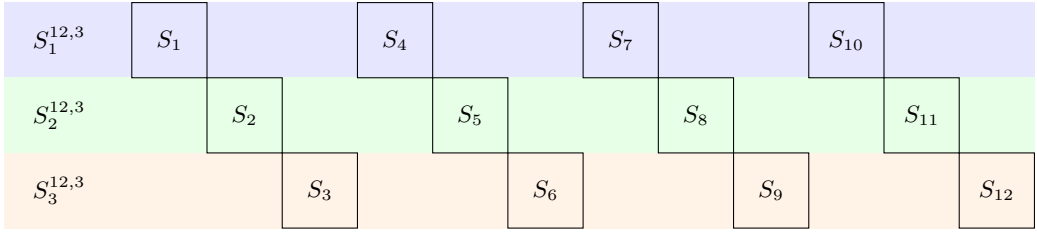
\begin{figure*}[t!]
\centering
\begin{tikzpicture}
\fill[blue!10] (-1.7,2) rectangle (3+3*2+3,3);
\fill[green!10] (-1.7,1) rectangle (3+3*2+3,2);
\fill[orange!10] (-1.7,0) rectangle (3+3*2+3,1);
\def\hspacing{2} 
\foreach \i in {0,1,2,3}{
  \draw (\i+\i*\hspacing,2) rectangle node[midway]{$S_{\the\numexpr\i*3+1\relax}$}(\i+\i*\hspacing+1,3);
  \draw (\i+\i*\hspacing+1,1) rectangle node[midway]{$S_{\the\numexpr\i*3+2\relax}$}(\i+\i*\hspacing+2,2);
  \draw (\i+\i*\hspacing+2,0) rectangle node[midway]{$S_{\the\numexpr\i*3+3\relax}$}(\i+\i*\hspacing+3,1);  
}

\node[left] at (-0.5,2.5) {$S_1^{12,3}$};
\node[left] at (-0.5,1.5) {$S_2^{12,3}$};
\node[left] at (-0.5,0.5) {$S_3^{12,3}$};
\end{tikzpicture}
\caption{Example of how the rounds are split up into interleaved groups for $n = 12$ and $s = 3$.}
\label{fig:round_grouping}
\end{figure*}

Under \cref{cond:eve_seq_blocks}, we can model Eve's attacks by a sequence of maps $\cA_i: E'_{i-1} Q_{i}^{i+s-1} \to E'_i Q_{i}^{i+s-1}$; the difference to the scenario in \cref{sec:eve_attack} is that now Eve can act on all of $Q_{i}^{i+s} = Q_i \cdots Q_{i+s-1}$ at the same time.
This prevents us from writing the final state $\rho$ as the output of a sequence of maps $\cM_i$ as in \cref{claim:eat_qkd_pm}.

To circumvent this issue, we will need to split the systems $S^n$ into interleaved groups of systems 
\begin{align*}
S_i^{n, s} \deq S_i S_{i + s} S_{i+2s} \cdots S_{i + \lfloor n/s \rfloor s} \,.
\end{align*}
This is illustrated in \cref{fig:round_grouping}.
We can bound the entropy of $S^n$ in terms of a sum of entropies of the individual groups $S_i^{n, s}$ for $i = 1, \dots, s-1$ by repeatedly applying the chain  rule for min-entropies~\cite{vitanov2013chain} for a total of $(s-1)$ number of times.
Setting $\eps' = \frac{\eps_s}{3s-2}$ as in \cref{claim:eat_block_pm}, we get that
\begin{align*}
&\hmin^{\eps_s}(S^n |  I^n C^n E'_n)_{\rho_{| \Omega_C}} \\
&= \hmin^{\eps_s}(S_1^{n, s} \dots S_s^{n, s} |  I^n C^n E'_n)_{\rho_{| \Omega_C}} \\
&\geq  \hmin^{\eps_s - 3 \eps'}(S_1^{n, s} \dots S_{s-1}^{n, s} |  I^n C^n E'_n)_{\rho_{| \Omega_C}} + \hmin^{\eps'}(S_s^{n, s} |  I^n C^n E'_n S_1^{n, s} \dots S_{s-1}^{n, s})_{\rho_{| \Omega_C}} - g(\eps') \\
&\!  \begin{multlined}
\geq\hmin^{\eps_s - 2 \cdot 3 \eps'}(S_1^{n, s} \dots S_{s-2}^{n, s} |  I^n C^n E'_n)_{\rho_{| \Omega_C}} \\ \qquad + \hmin^{\eps'}(S_{s-1}^{n, s} |  I^n C^n E'_n S_1^{n, s} \dots S_{s-2}^{n, s})_{\rho_{| \Omega_C}} + \hmin^{\eps'}(S_s^{n, s} |  I^n C^n E'_n S_1^{n, s} \dots S_{s-1}^{n, s})_{\rho_{| \Omega_C}} - 2 \cdot g(\eps') \end{multlined}\\
&\geq \dots \\
&\geq \hmin^{\eps_s - (s-1) \cdot 3 \eps'}(S_1^{n, s} |  I^n C^n E'_n)_{\rho_{| \Omega_C}} + \sum_{i = 2}^s \hmin^{\eps'}(S_i^{n, s} |  I^n C^n E'_n S_1^{n, s} \dots S_{i-1}^{n, s})_{\rho_{| \Omega_C}}  - (s-1) \cdot g(\eps') \\
&= \sum_{i = 1}^s \hmin^{\eps'}(S_i^{n, s} |  I^n C^n E'_n S_1^{n, s} \dots S_{i-1}^{n, s})_{\rho_{| \Omega_C}}  - (s-1) \cdot g(\eps') \,, \numberthis \label{eqn:total_entropy_groups}
\end{align*}
where the last line holds because $\eps_s - (s-1) \cdot 3 \eps' = \eps'$.
Having split the total entropy into such groups, we can now use the GEAT to bound each $\hmin^{\eps'}(S_i^{n, s} |  I^n C^n E'_n S_1^{n, s} \dots S_{i-1}^{n, s})_{\rho_{| \Omega_C}}$ in terms of single-round von Neumann entropies.
This works in exactly the same manner as the proof of \cref{claim:eat_qkd_pm} because by assumption, Eve cannot act simultaneously on more than one of the rounds corresponding to $S_i^{n, s} \deq S_i S_{i + s} S_{i+2s} \cdots S_{i + \lfloor n/s \rfloor s}$.
Therefore, for each individual group $S_i^{n, s}$, Eve's attack is a sequential attack \emph{on that group of rounds}.
As a result, following the same steps as in the proof of \cref{claim:eat_qkd_pm} (and assuming that $n$ is divisible by $s$ for simplicity) we get that for each group of rounds,
\begin{align*}
\hmin^{\eps'}(S_i^{n, s} |  I^n C^n E'_n S_1^{n, s} \dots S_{i-1}^{n, s})_{\rho_{| \Omega_C}} 
\geq \frac{n}{s} \, k_{\ca} - \frac{n}{s} \, \frac{\alpha-1}{2-\alpha} \, \frac{\ln(2)}{2} V^2 - \frac{g\left(\eps'\right) + \alpha \log(1/\pr{\Omega_C})}{\alpha-1} -  \frac{n}{s} \, \left( \frac{\alpha-1}{2-\alpha} \right)^2 K'(\alpha) \,.
\end{align*}
Note that the extra conditioning on $S_1^{n, s} \dots S_{i-1}^{n, s}$ does not make a difference to the proof as we may formally consider these systems as part of Eve's side information \emph{for that particular group of rounds}.
Inserting this bound into \cref{eqn:total_entropy_groups}, we get that 
\begin{align*}
\hmin^{\eps_s}(S^n |  I^n C^n E'_n)_{\rho_{| \Omega_C}}
\geq n \, k_{\ca} - n \, \frac{\alpha-1}{2-\alpha} \, \frac{\ln(2)}{2} V^2 - s \, \frac{g\left(\eps'\right) + \alpha \log(1/\pr{\Omega_C})}{\alpha-1} -  n \, \left( \frac{\alpha-1}{2-\alpha} \right)^2 K'(\alpha)  - (s-1) g \left( \eps' \right)
\end{align*}
as claimed in \cref{claim:eat_block_pm}.

Finally, we note that for certain parameter regimes one can improve the second-order terms in \cref{thm:block_qkd_pm} using exactly the same idea as above, but performing the splitting into interleaved groups of rounds at the level of Renyi entropies, not min-entropies.
Concretely, this means that instead of deriving \cref{eqn:total_entropy_groups}, one first relates $\hmin^{\eps_s}(S^n |  I^n C^n E'_n)_{\rho_{| \Omega_C}}$ to the Renyi entropy (see \cite[Definition 5.2]{tomamichel2015quantum}) $\hau(S^n |  I^n C^n E'_n)_{\rho_{| \Omega_C}}$, and then applies the chain rule from~\cite{dupuis2015chain} according to a binary tree of depth $O(\log s)$,  i.e.~on the first application of the chain rule one splits the rounds into two equally sized groups (corresponding to a step size of 2), on the second application one again splits each of these groups into equally sized subgroups (corresponding to a step size of 4), and so on, until after $O(\log s)$ repetitions the desired step size $s$ is reached.
The remainder of the analysis is then identical.

\begin{figure}[t!]
\centering
\begin{subfigure}[T]{0.49\linewidth}
\centering
    \begin{tikzpicture}
	\begin{axis}[
		height=6cm,
		width=8.3cm,
		xlabel=Depolarising probability $p$ in $\%$,
		ylabel=Key rate,
		xmin=0,
		xmax=0.095,
		ymax=0.25,
		ymin=0,
	     xtick={0,0.03,0.06,0.09,0.12,0.15},
	     xticklabels={0, 3, 6, 9, 12, 15},
          ytick={0,0.05,0.1,0.15,0.2,0.25},
          yticklabels={0,0.05,0.1,0.15,0.2,0.25},
          reverse legend,
		legend style={legend cell align=left,font=\footnotesize} 
	]
\input{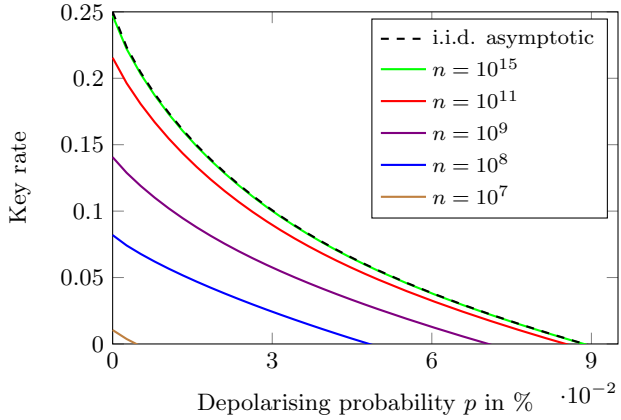}
\end{axis}
\end{tikzpicture}
\caption{Step size $s = 10$.\label{fig:block10}}
\end{subfigure}
\begin{subfigure}[T]{0.49\linewidth}
\centering
\begin{tikzpicture}
	\begin{axis}[
		height=6cm,
		width=8.3cm,
		xlabel=Depolarising probability $p$ in $\%$,
		ylabel=Key rate,
		xmin=0,
		xmax=0.095,
		ymax=0.25,
		ymin=0,
	     xtick={0,0.03,0.06,0.09,0.12,0.15},
	     xticklabels={0, 3, 6, 9, 12, 15},
          ytick={0,0.05,0.1,0.15,0.2,0.25},
          yticklabels={0,0.05,0.1,0.15,0.2,0.25},
          reverse legend,
		legend style={legend cell align=left,font=\footnotesize} 
	]
	
\input{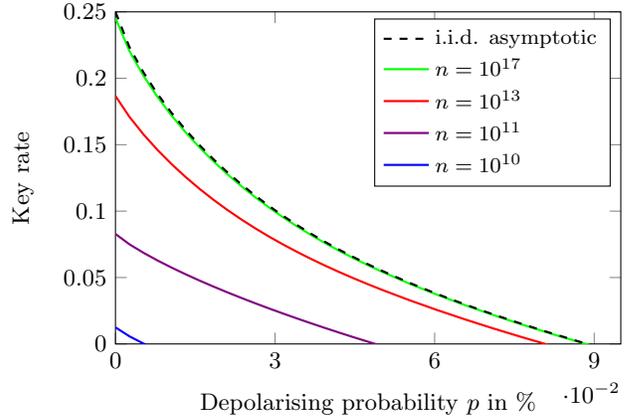}
\end{axis}
\end{tikzpicture}
\caption{Step size $s = 10^4$. Note that the same colors here do not correspond to the same $n$ as in (a).\label{fig:block1e4}}
\end{subfigure}
\caption{B92 example with weakened sequentiality condition.
We consider the same B92 protocol with the same parameters as in the main text (\cref{sec:b92}, see in particular~\cref{fig:b92_curves}), except that we weaken the sequentiality condition to allow for some step size $s$, and as a result have to use \cref{thm:block_qkd_pm} to obtain the key rates.
The step size $s = 10$ is realistic for satellite-to-earth QKD experiments~\cite{liao2017satellite},\textsuperscript{a} whereas the plot for $s = 10^{4}$ illustrates what happens at larger block sizes that are more relevant for fibre-based QKD implementations.
We see that while the asymptotic key rate remains the same irrespective of block size, at larger block sizes finite-size corrections become more relevant.
We note that the key rates are mostly included for illustrative purposes and could be improved further by using the Renyi-based method explained at the end of this section.}
\label{fig:b92_curves_blocked}
\vspace{0.2cm}
\textsuperscript{a} \footnotesize 
For this example, we consider the worst case, where Eve is able to speed up the signals to the speed of light in vaccuum.
This defines the largest possible causal future in \cref{fig:spacetime} that is still compatible with special relativity.
Using the effective thickness of the atmosphere $d \approx 8$km~\cite{hom_atmos} and letting $n_{\rm Air} \approx 1.0003$ be the refractive index of air, if the signal is sent at a 45 degree angle through the atmosphere, the delay of the signal compared to one travelling at the speed of light $c$ is $\Delta t = \sqrt{2} (n_{\rm Air} - 1)d/c \approx 10^{-8}s$.
Assuming a signal frequency $f_{\rm Signal} = 1$ GHz (which exceeds the one used e.g.~by~\cite{liao2017satellite}), we see that a step size $s = \Delta t \cdot f_{\rm Signal} \approx 10$ is sufficient.
\end{figure}

\section{Entanglement-based protocols} \label{sec:entanglement_based}

Our general entanglement-based QKD protocol is very similar to the prepare-and-measure protocol in \cref{sec:pm_general}. 
The only difference is in the the data generation step: in \cref{prot:qkd-pm}, Alice prepared a state $\psi_{UQ}$ and Bob measured system $Q$, storing his outcome in register $V$.
In contrast, in \cref{prot:qkd-ent}, Eve prepares a state $\psi_{PQE}$, and sends $P$ to Alice and $Q$ to Bob. Then, Alice and Bob measure their respective systems, recording the outcomes in registers $U$ and $V$.
The raw data in $U$ and $V$ is then treated exactly the same as in \cref{prot:qkd-pm}.
Even though \cref{step_ent:pd}-\cref{step_ent:pa} are identical to \cref{prot:qkd-pm}, we spell out the full protocol for reference.

\begin{longfbox}[breakable=true, padding=1em, padding-right=1.8em, padding-top=1.2em, margin-top=1em, margin-bottom=1em]
\begin{protocol} {\bf General entanglement-based QKD protocol} \label{prot:qkd-ent} \end{protocol}
\noindent\underline{Protocol arguments}  \vspace{-0.8ex}
\vspace{-0.6ex}
\begin{center}
\begin{tabularx}{\textwidth}{r c X}
$n \in \N$ &:& number of rounds \\
$\{M^{(u)}\}_{u \in \cU}$, $\{N^{(v)}\}_{v \in \cV}$ &:& POVMs acting on Hilbert spaces $\H_P, \H_Q$, respectively, describing Alice's and Bob's measurements with $\cU$ and $\cV$ the set of possible outcomes \\
$\pd: \cU \times \cV \to \cI$ &:& function describing transcript of public discussion (where $\cI$ is some finite alphabet) \\
$\rk: \cU \times \cI \to \cS$ &:& function describing Alice's raw key generation (where $\cS$ is the alphabet of the raw key) \\
$\ev: \cV \times \cI \times \cS \to \cC$ &:& function ``evaluating'' each round by assigning a label from the alphabet $\cC$ \\ 
$k_{\ca} > 0$ &:& required amount of single-round entropy generation \\
$\eps_{\kv}, \eps_{\pa} > 0$ &:& tolerated errors during key validation and privacy amplification steps \\
$\ca: \mbP(\cC) \to \R$ &:& function corresponding to collective attack bound \\
$l \in \N$ &:& length of final key \\
\end{tabularx}
\end{center}

\vspace{0.6ex}

\noindent\underline{Protocol steps}
\begin{enumerate}[label=(\arabic*)]
    \setlength{\itemsep}{0.4ex}
    \setlength{\parskip}{0.4ex}
    \setlength{\parsep}{0.4ex} 
\item \emph{Data generation.} Alice receives systems $P^n$ and Bob systems $Q^n$ of an initial quantum state $\psi_{P^n Q^n E}$ prepared by Eve.
For each $i \in \{1, \dots, n\}$, Alice measures the POVM $\{M^{(u)}\}_{u \in \cU}$ on register $P_i$ of the state $\psi_{P^n Q^n E}$ and records the outcome in register $U_i$. Similarly, Bob measures $\{N^{(v)}\}_{v \in \cV}$ on register $Q_i$ and records the outcome in register $V_i$.
\item \emph{Public discussion.}
For each $i \in \{1, \dots, n\}$, Alice and Bob publicly exchange information $I_i = \pd(U_i, V_i)$. \label{step_ent:pd}
\item \emph{Raw key generation.} For each $i \in \{1, \dots, n\}$, Alice computes $S_i = \rk(U_i, I_i)$. \label{step_ent:raw_key_gen}
\item \emph{Error correction.} Alice and Bob publicly exchange information $\ec \in \bits^{\lambda_{\ec}}$, which can depend on $U^n, V^n$, and $I^n$. Bob computes $\hat S^n(\ec, V^n, I^n) \in \cS^n$. \label{step_ent:ec}
\item \emph{Raw key validation.} Alice chooses a function $\hash: \cS^n \to \bits^{\lceil \log(1/\eps_{\kv}) \rceil}$ from a universal hash family $\cF$ (\cref{def:univ_hash}) according to the associated probability distribution $P_{\cF}$ and publishes a description of $f$ and the value $\hash(S^n)$. \label{step_ent:rkv}
Bob computes $\hash(\hat S^n)$ and aborts the protocol if $\hash(S^n) \neq \hash(\hat S^n)$.
\item \emph{Statistical check.} For each $i \in \{1, \dots, n\}$, Bob sets $\hat C_i = \ev(V_i, I_i, \hat S_i)$. Bob then computes $\ca(\freq{C^n})$. If the result is less than $k_\ca$, he aborts the protocol. \label{step_ent:stat}
\item \emph{Privacy amplification.} Alice and Bob convert their registers $S^n$ and $\hat S^n$ to a binary representation, obtaining strings of length $m$.
Alice chooses a seed $\mu \in \bits^m$ uniformly at random and publishes her choice.
Alice and Bob compute $l$-bit strings $K = \ext(S^n, \mu)$ and $\hat K = \ext(\hat S^n, \mu)$, respectively, and $\ext: \bits^m \times \bits^m \to \bits^l$ is a quantum-proof strong $(l + \lceil 2 \log(1/\eps_{\pa})\rceil, \eps_{\pa})$-extractor (\cref{def:extractor}). \label{step_ent:pa}
\end{enumerate} 
\end{longfbox}

The rest of this section proceeds similarly to \cref{sec:pm_general}: we first explain how to model Eve's attack, again distinguishing between general and collective attacks.
We formally define collective attack bounds for \cref{prot:qkd-ent} in \cref{def:coll_attack_ent}.
Then, in \cref{thm:general_qkd_ent}, we analyse the security of~\cref{prot:qkd-ent} assuming a collective attack bound.
The definitions and the security proof are very similar to \cref{sec:pm_general}, so we give less detailed explanations and only point out the relevant differences for the proof.

Eve's attack in \cref{prot:qkd-ent} is specified by her choice of the state $\psi_{P^n Q^n E}$.
An honest Eve would distribute some desired product state $\hat \psi$, e.g.~an EPR pair, and keep no side information, i.e.~$\psi_{P^n Q^n E} = \hat \psi_{PQ}^{\ot n}$.
The most general attack available to Eve consists in preparing an arbitrary state $\psi_{P^n Q^n E}$.
We note that Eve's attack in an entanglement-based protocol is not subject to a sequentiality condition (\cref{cond:eve_seq}) as in a prepare-and-measure protocol.
This is because in an entanglement-based protocol, Eve's attack occurs at the level of the input state, which can be arbitrary, and the actions in the protocol performed by Alice and Bob are sequential irrespective of Eve's choice of input state.
In contrast, in a prepare-and-measure protocol, Eve's attack is part of the actions performed \emph{during} the protocol and therefore needs to be modelled as part of the quantum channels applied during the protocol; as a result, for the protocol as a whole to still have a sequential structure, Eve's attack needs to have such a structure, too.
This is not an artefact of the GEAT, but rather a structural difference between entanglement-based and prepare-and-measure protocols.

As in \cref{sec:pm_general}, a collective attack is the special case where Eve behaves in an i.i.d.~manner, i.e.~Eve prepares a product state $\psi_{P^n Q^n E} = \psi_{P Q E}^{\ot n}$ for some arbitrary state $\psi_{P Q E}$.
Formally, we can define a collective attack bound for \cref{prot:qkd-ent} similarly to \cref{def:coll_attack}.

\begin{definition}[Collective attack bound for \cref{prot:qkd-ent}] \label{def:coll_attack_ent}
Fix arguments $\{M^{(u)}\}_{u \in \cU}$, $\{N^{(v)}\}_{v \in \cV}$, $\pd$, $\rk$, and $\ev$ for \cref{prot:qkd-ent}.
Suppose that Alice and Bob run a single round (i.e.~$n=1$) of  \cref{prot:qkd-ent} up to (and including) \cref{step_ent:raw_key_gen}.\footnote{Note that for this, the other arguments specified in the description of \cref{prot:qkd-ent} are not required, so we do not need to specify them. The collective attack bound only depends on the protocol arguments specified in \cref{def:coll_attack_ent}.}
For a choice of Eve's state $\psi_{P Q E}$, denote the state at the end of \cref{step_ent:raw_key_gen} as $\nu_{UVSIE}$.
Let $\nu_{UVSIEC}$ be an extension of this state, where $C = \ev(V, I, S)$.
A collective attack bound (for the choice of parameters fixed above) is a map $\ca: \mbP(\cC) \to \R$ such that for any initial state $\psi_{PQE}$ prepared by Eve, the state $\nu_{CUVSIE}$ satisfies 
\begin{align*}
H(S | I E C)_{\rho} \geq \ca(\rho_{C}) \,.
\end{align*}
\end{definition}
It is easy to see that for states that minimize the l.h.s.~of this inequality, the system $E$ is a purification of $P$ and $Q$.
Hence, it suffices to restrict our attention to such states.
We are now ready to prove the security statement for \cref{prot:qkd-ent}.

\begin{theorem} \label{thm:general_qkd_ent}
Fix any choice of arguments $n$, $\{M^{(u)}\}_{u \in \cU}$, $\{N^{(v)}\}_{v \in \cV}$, $\pd$, $\rk$, $\ev$, $k_\ca$, $\eps_{\kv}$, and $\eps_{\pa}$ for \cref{prot:qkd-ent}.
Let $\ca: \cC \to \R$ be an affine collective attack bound for this choice of arguments.
For any $\eps_s, \eps_a > 0$ and $\alpha \in (1, 3/2)$\footnote{These are parameters that can be optimised to maximize the key length}, choose a final key length $l$ that satisfies
\begin{multline*}
l \leq n \, k_{\ca} - n \, \frac{\alpha-1}{2-\alpha} \, \frac{\ln(2)}{2} V^2 - \frac{g({\eps_s}) + \alpha \log(1/\eps_a)}{\alpha-1} -  n \, \left( \frac{\alpha-1}{2-\alpha} \right)^2 K'(\alpha) \\
- \lceil 2 \log(1/\eps_{\pa})\rceil - \lambda_{\ec} - \lceil \log(1/\eps_{\kv}) \rceil \,,
\end{multline*}
where
\begin{align*}
g({\eps_s}) &= \log(1 - \sqrt{1-{\eps_s}^2}) \,,\\
V &= \log (2|\cS|^2+1) + \sqrt{2 + \Var{\ca}} \,,\\
K'(\alpha) &= \frac{(2-\alpha)^3}{6 (3-2\,\alpha)^3 \ln 2} \, 2^{\frac{\alpha-1}{2-\alpha}(2\log |\cS|  + \Max{\ca} - \MinSigma{\ca})} \ln^3\left( 2^{2\log |\cS| + \Max{\ca} - \MinSigma{\ca}} + e^2 \right) \,.
\end{align*}
With this choice of parameters, \cref{prot:qkd-ent} is $\eps^{\setft{cor}}$-correct and $\eps^{\setft{sec}}$-secret for 
\begin{align*}
\eps^{\setft{cor}} = \eps_{\kv} \,, \qquad \eps^{\setft{sec}} = \max\{\eps_{\pa} + 4 \, \eps_s, 2 \eps_a\} + 2 \, \eps_{\kv}\,.
\end{align*}
\end{theorem}

\begin{proof}
The correctness statement is analogous to the proof of \cref{thm:general_qkd_pm}, so we focus on the secrecy condition.
Alice, Bob, and Eve's joint final state at the end of the protocol is denoted by $\rho_{U^nV^nI^nS^n\hat S^n \hat C^n K \hat K E}$.
As in the proof of \cref{thm:general_qkd_pm} add to $\rho$ additional systems $C^n$ defined by
\begin{align}
C_i = \ev(V_i, I_i, S_i) \label{eqn_ent:c_def}
\end{align}
and define the event $\Omega_{C}$ by the condition $\ca(\freq{C^n}) \geq k_{\ca}$.
Then, we can follow the same steps as in the proof of \cref{thm:general_qkd_pm} to reduce the \cref{thm:general_qkd_ent} to the following \cref{claim:eat_qkd_ent}.
\end{proof}

\begin{claim}\label{claim:eat_qkd_ent}
Continuing with the notation from before, for any $\alpha \in (1, 3/2)$:
\begin{align}
\hmin^{\eps_s}(S^n |  I^n C^n E)_{\rho_{|\Omega_C}}
\geq n \, k_{\ca} - n \, \frac{\alpha-1}{2-\alpha} \, \frac{\ln(2)}{2} V^2 - \frac{g({\eps_s}) + \alpha \log(1/\prs{\rho}{\Omega_C})}{\alpha-1} -  n \, \left( \frac{\alpha-1}{2-\alpha} \right)^2 K'(\alpha)\, , 
\end{align}
where $g({\eps_s})$, $V$, and $K'(\alpha)$ are defined as in \cref{thm:general_qkd_ent}.
\end{claim}
\begin{proof}
To make us of the GEAT, we need to write $\rho_{S^n I^n C^n E | \Omega_C}$ as the result of repeatedly applying quantum channels $\cM_1, \dots, \cM_n$ to Eve's (arbitrary) initial state $\psi_{P^n Q^n E}$ in \cref{prot:qkd-ent}.
For this, we define
\begin{align*}
\cM_i : P_i Q_i \to S_i I_i C_i
\end{align*}
as the following channel: given a quantum system $\omega_{P_i Q_i}$, \begin{enumerate}
\item measure the POVMs $\{M^{(u)}\}_{u \in \cU}$ and $\{N^{(v)}\}_{v \in \cV}$ on $P_i$ and $Q_i$ respectively, and store the results in registers $U_i$ and $V_i$,
\item set $I_i = \pd(U_i, V_i)$,
\item set $S_i = \rk(U_i, I_i)$,
\item set $C_i = \ev(V_i, I_i, S_i)$,
\item trace out registers $U_i$ and $V_i$.
\end{enumerate}
Comparing the steps of \cref{prot:qkd-ent} and \cref{eqn_ent:c_def} with this definition of $\cM_i$, we see that the marginal of $\rho$ on systems $S^n I^n C^n E$ is the same as the output of the maps $\cM_i$: 
\begin{align*}
\rho_{S^n I^n C^n E} = \cM_n \circ \dots \circ \cM_1(\psi_{P^n Q^n E}) \,.
\end{align*}
If we define the systems $E_i = P_{i+1}^n Q_{i+1}^n I^i C^i E$, then by suitable tensoring with the identity map and copying the register $C_i$ we can also view $\cM_i$ as a map 
\begin{align*}
\tilde \cM_i: E_{i-1} \to S_i E_i C_i \,.
\end{align*}
Then, we can also express the final state (which technically now includes two copies of $C^n$) as
\begin{align*}
\rho_{S^n E_n C^n} = \tilde \cM_n \circ \dots \circ \tilde \cM_1(\psi_{P^n Q^n E}) \,.
\end{align*}

To apply \cref{thm:with_testing}, we first need to check that the required conditions on the maps $\tilde \cM_i$ are satisfied.
The condition \cref{eqn:measurement_condition} is clearly satisfied as the systems $C_i$ are themselves included in the conditioning system.
The no-signalling condition in \cref{thm:with_testing} is also trivially satisfied in this case since there is no system $R_i$.

We now want to argue that the collective attack bound $\ca: \mbP(\cC) \to \R$ used as an argument in the protocol is a min-tradeoff function for the maps $\{\cM_i\}$.
By \cref{def:tradeoff}, for this we need to show that for any $i$ and any state $\omega^{i-1}_{E_{i-1} \tilde E_{i-1}}$ (where $\tilde E_{i-1} \equiv E_{i-1}$):
\begin{align}
\ca(\cM_i(\omega^{i-1})_{C_i}) \leq H(S_i | E_i \tilde E_{i-1})_{\tilde \cM_i(\omega^{i-1})} = H(S_i | I_i C_i P_{i+1}^n Q_{i+1}^n I^{i-1} C^{i-1} E \tilde E_{i-1})_{\tilde \cM_i(\omega^{i-1})} \,. \label{eqn_ent:coll_att_mintradeoff}
\end{align}
Remembering that $\tilde \cM_i$ acts as identity on the systems $P_{i+1}^n Q_{i+1}^n I^{i-1} C^{i-1} E \tilde E_{i-1}$, we can consider these systems collectively as a purifying system.
Since \cref{def:coll_attack_ent} allows arbitrary purifying systems, we see that \cref{eqn_ent:coll_att_mintradeoff} holds for any collective attack bound $\ca$ and  conclude that $\ca$ is a min-tradeoff function for  $\{\cM_i\}$.
By definition, for any $c^n \in \Omega_{C}$, $\ca(\freq{c^n}) \geq k_{\ca}$.
Therefore, \cref{claim:eat_qkd_ent} follows by applying \cref{thm:with_testing}.
\end{proof}

\section{Simplification of  the optimisation problem for the B92 protocol} \label{app:b92_simpler}
Here we describe additional simplifications to the numerical optimisation problem from \cref{eqn:c_opt_psi} for the case of the B92 protocol as it is specified in \cref{sec:b92}.
As a first simplification, we exploit the fact that Alice and Bob distinguish between ``test rounds'', where $T = 1$ and they use the function $\ev_{T = 1}$, and ``data rounds'', where $T = 0$ and they perform no statistical check.
We can therefore split the state $\nu$ as 
\begin{align*}
\nu = (1 - \gamma) \nu^{(\setft{data})} + \gamma \, \nu^{\setft{(test)}} \,,
\end{align*}
where $\nu^{(\setft{data})}_C = \proj{\bot}$ and $\nu^{(\setft{test})}_C$ is a distribution over $\cC' = \{\texttt{fail}, \texttt{inc}, \varnothing\}$ determined according to $\ev_{T = 1}$.
We can now apply \cite[Lemma V.5]{dupuis2019entropy}, which states the following (translated to our notation):
if an affine function $g: \mbP(\cC') \to \R$ satisfies that for any initial state $\hat \psi_{PQ}$, 
\begin{align}
g(\nu^{(\setft{test})}_C) \leq H(S|IEC)_{\nu} \,, \label{eqn:g_prop}
\end{align}
then the affine function $\ca: \cC \to \R$ for $\cC = \cC' \cup \{\bot\}$ defined by 
\begin{align*}
\ca(\delta_c) &\deq \Max{g} + \frac{1}{\gamma} (g(\delta_c) - \Max{g}) \quad \setft{for } c \in \cC' \,\\
\ca(\delta_\bot) &\deq \Max{g} 
\end{align*}
is a collective attack bound.
Here, $\Max{g}$ is defined as in \cref{eqn:def_var_f} and $\delta_c$ denotes the point distribution with all weight on element $\delta_c$.
To evaluate $\ca$ on any distribution $\nu_C$, one simply writes that distribution as a convex combination of such point distributions and uses that $\ca$ is affine, i.e.~linear under convex combinations.
In addition, \cite[Lemma V.5]{dupuis2019entropy} also provides simple formulae for the properties of $\ca$ as in \cref{eqn:def_var_f} in terms of the properties of $g$.
The main advantage of this approach over a direct evaluation of the optimisation problem from \cref{eqn:c_opt_psi} is that for small values of $\gamma$, the latter often runs into numerical stability issues, whereas the former does not.

The problem of finding a collective attack bound is therefore reduced to finding a function $g$ that satisfies \cref{eqn:g_prop}.
This can be achieved using the same method as in \cref{sec:col_attack_bounds}.
We make the ansatz $g(\nu^{(\setft{test})}_C) = \vec \lambda' \cdot \vec \nu^{(\setft{test})}_C + c_{\vec \lambda'}$ and choose $\vec \lambda'$ heuristically, e.g.~using Matlab's $\texttt{fminsearch}$.
Given a choice of $\vec \lambda'$, we need to determine $c_{\vec \lambda'}$ such that $g$ satisfies \cref{eqn:g_prop}.
Following the steps of \cref{sec:col_attack_bounds}, we can see that a valid choice of $c_{\vec \lambda'}$ is given by the solution to the following convex optimisation problem:
\begin{align*}
c_{\vec \lambda'} =
&\inf_{\hat \psi_{PQ}} \dnormal{\nu^1_{PQSIC}}{\cP_S(\nu^1_{PQSIC})} - \vec \lambda' \cdot \vec \nu^{(\setft{test})}_C \\
& \sth \hat \psi_{PQ} \geq 0\,, \quad \tr{\hat \psi_{PQ}} = 1\,,\quad\hat \psi_P = \tilde \psi_P\,, 
\end{align*}
where $\nu^1$ is as defined in \cref{eqn:def_nu1}.
We can simplify this optimisation problem by focussing on the case $I = \top$ and $C = \bot$ since we only expect to generate raw data for the secret key in a data round whose measurement was not inconclusive:\footnote{Note that the following inequality always holds by data processing and because $I$ and $C$ is are classical registers, so we always obtain a valid (albeit possibly less tight) key rate after applying it. The preceding explanation argues that in the case of B92, this inequality is (close to) an equality and therefore using it does not significantly lower the key rate.}
\begin{align*}
\dnormal{\nu^1_{PQSIC}}{\cP_S(\nu^1_{PQSIC})} \geq \dnormal{\nu^1_{PQS\wedge I = \top \wedge C = \bot}}{\cP_S(\nu^1_{PQS \wedge I = \top \wedge C = \bot})} \,.
\end{align*}
Computing $M^{(s,i,c)}_{PQ}$ according to \cref{eqn:def_M}, inserting this into \cref{eqn:def_nu1}, and using that the square root acts trivially on projectors, we can write $\nu^1_{PQS \wedge I = \top \wedge C = \bot}$ explicitly as 
\begin{align*}
\nu^1_{PQS \wedge I = \top \wedge C = \bot} = (1 - \gamma) \sum_{s, s' \in \bits} \Big(\proj{s}_P \ot \sqrt{\1 - N^{(\bot)}_Q}\Big) \hat \psi_{PQ} \Big(\proj{s'}_P \ot \sqrt{\1 - N^{(\bot)}_Q}\Big) \ot \ket{s}\!\bra{s'} \,.
\end{align*}
This state arises from
\begin{align*}
\tilde \nu_{PQ \wedge I = \top \wedge C = \bot} = (1 - \gamma)  \Big(\1_P \ot \sqrt{\1 - N^{(\bot)}_Q}\Big) \hat \psi_{PQ} \Big(\1_P \ot \sqrt{\1 - N^{(\bot)}_Q}\Big)
\end{align*}
by application of the isometry $V = \sum_{s \in \bits} \proj{s}_P \ot \ket{s}_S$.
Similarly we can see that 
\begin{align*}
\cP_S(\nu^1_{PQS \wedge I = \top \wedge C = \bot}) = V \cP_P(\tilde \nu_{PQ \wedge I = \top \wedge C = \bot}) V^\dagger \,,
\end{align*}
so by isometric invariance of the relative entropy:
\begin{align*}
\dnormal{\nu^1_{PQS\wedge I = \top \wedge C = \bot}}{\cP_S(\nu^1_{PQS \wedge I = \top \wedge C = \bot})} = \dnormal{\tilde \nu_{PQ\wedge I = \top \wedge C = \bot}}{\cP_P(\tilde \nu_{PQ \wedge I = \top \wedge C = \bot})} \,.
\end{align*}
Therefore, we can find a valid choice for $c_{\vec \lambda}$ by solving the simplified optimisation problem 
\begin{align*}
c_{\vec \lambda'} = &\inf_{\hat \psi_{PQ}} \dnormal{\tilde \nu_{PQ\wedge I = \top \wedge C = \bot}}{\cP_P(\tilde \nu_{PQ \wedge I = \top \wedge C = \bot})} - \vec \lambda' \cdot \vec \nu^{(\setft{test})}_C \\
& \sth \hat \psi_{PQ} \geq 0\,, \quad \tr{\hat \psi_{PQ}} = 1\,,\quad\hat \psi_P = \tilde \psi_P\,.
\end{align*}
We solve this optimisation problem using the package \texttt{CVXQUAD}~\cite{fawzi2019semidefinite} and thus obtain the desired affine collective attack bound.

\section{Parameters for B92 example} \label{sec:b92_params}

In this section, we detail our choice of parameters for the B92 protocol presented in the main text.
Most parameter choices are summarised in the following table.

\begin{center}
\begin{tabularx}{12cm}{|c|c|X|}
\hline
Symbol & Value & Description \\
\hline 
$n$ & see \cref{fig:b92_curves} & number of rounds \\
$\eps_{\kv}$ & $5\cdot10^{-11}$ & tolerated error during key validation step \\
$\eps_\pa$ & $10^{-10}$& tolerated error during privacy amplification step \\
$\eps_a$ & $4 \cdot 10^{-10}$& minimum probability of passing statistical check \\
$\eps_s$ & $2 \cdot 10^{-10}$& smoothing parameter \\
$\eps^{\setft{comp}}_\kv$ & $5 \cdot 10^{-3}$ & abort probability from key validation step \\
$\eps^{\setft{comp}}_\ev$ & $5 \cdot 10^{-3}$ & abort probability from statistical check\\
\hline
\end{tabularx}
\end{center}

We refer to \cref{prot:qkd-pm} and \cref{thm:general_qkd_pm} for a more detailed description of the role of each parameter.
The testing probability $\gamma$ and the parameter $\alpha$ from \cref{thm:general_qkd_pm} are optimised numerically.
We choose the remaining arguments for \cref{prot:qkd-pm} so that they satisfy the conditions in \cref{lem:comp_ec}, \cref{lem:comp_ev}, and \cref{thm:general_qkd_pm} and compute the resulting values of $\eps^{\setft{cor}}$, $\eps^{\setft{sec}}$, and $\eps^{\setft{comp}}$.
Specifically, to choose a value for $\lambda_{\ec}$, a direct calculation shows that for noise level $p$, we have $H(S | V I)_{\nu^{\setft{hon}}} = \frac{1+p}{4} h(1/(1+p))$, where $h(x) = - x \log x - (1-x) \log (1-x)$ is the binary entropy, and $\nu^{\setft{hon}}$ is the state for an honest noisy implementation and depends implicitly on the noise level $p$.
We then set 
\begin{align*}
\lambda_{\ec} = \left\lceil n \frac{1+p}{4} h(1/(1+p)) + 2 \sqrt{n} \sqrt{1 - 2 \log (\eps^{\setft{comp}}_\kv / 2)} \log(7) + 2 \log \frac{2}{\eps^{\setft{comp}}_\kv} \right\rceil \,.
\end{align*}
We furthermore choose $k_{\ca} = \ca(\nu_C^{\setft{hon}}) - \delta$ for
\begin{align*}
\delta = \left( \frac{2 (\Max{\ca} - \Min{\ca}) \ca(\nu_C^{\setft{hon}}) + 6 \Var{\ca}}{3 n} \log \frac{1}{\eps^{\setft{comp}}_\ev - \eps_{\kv}}  \right)^{1/2} \,.
\end{align*}
It is easy to see that the above choices satisfy the conditions in \cref{lem:comp_ec} and \cref{lem:comp_ev},\footnote{Note that we have replaced the factor of $\delta$ on the r.h.s.~in \cref{lem:comp_ev} by $\ca(\nu_C^{\setft{hon}})$. This can be done since we are only interested in cases where $k_\ca > 0$, i.e.~we can assume $\delta \leq \ca(\nu_C^{\setft{hon}})$.} whence it follows that the total completeness error for this choice of parameters is 
\begin{align*}
\eps^{\setft{comp}} = \eps^{\setft{comp}}_\kv + \eps^{\setft{comp}}_\ev = 10^{-2} \,.
\end{align*}

Finally, we choose the key length to be the largest integer $l$ that satisfies the condition in \cref{eqn:key_length}.
Then, we can apply \cref{thm:general_qkd_pm} to find that the B92 protocol with this choice of parameters is $\eps^{\setft{cor}}$-correct and $\eps^{\setft{sec}}$-secret with
\begin{align*}
\eps^{\setft{cor}} = \eps_{\kv} = 5 \cdot 10^{-11}\,, \qquad \eps^{\setft{sec}} = \max\{\eps_{\pa} + 4 \, \eps_s, 2\,\eps_a\} + 2 \, \eps_{\kv} \leq 10^{-9} \,.
\end{align*}

\section{BB84 protocol with decoy states} \label{app:decoy}

The BB84 protocol~\cite{bb84} is the most well-known QKD protocol and we already described how it can be viewed as an instance of \cref{prot:qkd-pm} in \cref{sec:gen_pm_protocol}.
That description assumed that Alice always sends a single qubit to Bob.
This qubit could e.g.~be implemented as a polarised photon.
However, in practice Alice will usually use a highly attenuated laser that does not reliably output a single photon.
Instead, the number $s$ of photons in a laser pulse is distributed according to the Possonian distribution $p_L(s|\mu) = e^{-\mu} \frac{\mu^s}{s!}$, where the average photon number $\mu$ depends on the laser's intensity and is known to Alice.
This means that a single pulse may contain multiple photons (with the same polarisation), allowing Eve to perform a photon number splitting attack~\cite{lutkenhaus2002quantum}: Eve measures the number of photons in a pulse and, if there are multiple photons, measures the polarisation of one of the photons and forwards the others to Bob unchanged; this tells Eve the corresponding bit in the raw key.
Therefore, only raw key bits from rounds with exactly one photon in the laser pulse contribute to the secret key.
We therefore need to lower-bound the fraction $\Omega$ of ``single-photon rounds'' and the error rate $e_1$ in those rounds.
While such a bound can be obtained from the standard BB84 protocol, the resulting key rate is quite poor.
The decoy state method is a modification of the BB84 protocol that allows for better estimation of $\Omega$ and $e_1$ and therefore achieves higher key rates; see \cite{hwang2003quantum,lo2005decoy,ma2005practical,lim2014concise} for a background on the idea of decoy state protocols.

Decoy-state protocols can be implemented with commercially available components and thus serve as good examples for practical QKD. The purpose of this section is to show that the methods developed here can be used to establish their security.
For this, we give a formal description of the BB84 decoy state protocol as an instance of \cref{prot:qkd-pm} and explain how an existing i.i.d.~asymptotic analysis of the BB84 decoy state protocol can be understood as a collective attack bound in our framework.
This puts us in a position to apply \cref{thm:general_qkd_pm}: we simply need to compute the relevant properties of the collective attack bound and numerically optimize over the parameters of the protocol to obtain the maximum key rate.
As our goal is to illustrate the use of our framework, and not to numerically optimize key rates, we leave a detailed numerical analysis with experimentally realistic noise models for future work.
This example again highlights the ease of use of our framework: one only needs to verify that the protocol fits into the template \cref{prot:qkd-pm} and can reuse collective attack bounds from prior work to immediately obtain a finite-size key rate against general attacks.
In contrast to previous techniques such as the de Finetti theorem, using photonic protocols over qubit protocols introduces no additional complications because our \cref{thm:general_qkd_pm} only depends on the collective attack bound, not the underlying quantum states.

We now describe the BB84 decoy state protocol as an instance of \cref{prot:qkd-pm}.\footnote{We note that there are many different variants of the decoy state BB84 protocol. Here, we use the version described in~\cite{lim2014concise}, except that we use a fixed number of total rounds.}
We denote by $\psi_{\mu, x, a}$ the state produced by a laser pulse at intensity $\mu$ with (polarisation) basis $x \in \zx$ and value (polarisation direction) $a \in \bits$.
In the protocol, Alice chooses an intensity $\mu \in \{\mu_1, \mu_2, \mu_3\}$ with probabilities $p_{\mu_1}, p_{\mu_2}, p_{\mu_3}$, respectively, for some $\mu_i$ satisfying $\mu_1 > \mu_2 + \mu_3$ and $\mu_2 > \mu_3 \geq 0$.
Alice also chooses a basis $x \in \zx$ with probability $q_x$, and a value $a \in \bits$ uniformly at random. 
Therefore, 
\begin{align*}
\psi_{UQ} = \sum_{\substack{\mu \in \{\mu_1, \mu_2, \mu_3\} \\ x \in \zx}} p_{\mu} q_x \proj{\mu, x, a}_U \ot \left( \psi_{\mu, x, a} \right)_Q.
\end{align*}

Bob will choose a measurement basis $y \in \zx$ with probability $q_y = q_x$ and measure the signal he receives from Alice.
The measurement will yield an outcome $b \in \{\emptyset, 0, 1\}$, where $\emptyset$ corresponds to Bob not detecting any photon and $0, 1$ denote Bob's measured polarisations.\footnote{In practice, this is implemented using a polarising beam splitter and two detectors. If no detector clicks, Bob records $\emptyset$; if the first (second) detector clicks, Bob record 0 (1); if both detectors click, Bob chooses $b \in \zx$ uniformly at random.}
This measurement can be described by a POVM $\{N^{(y, b)}\}$.
Because we will reuse a collective attack bound from~\cite{lim2014concise} instead of deriving our own, there is no need to write out this POVM explicitly.

During public discussion,  Alice announces the intensity $\mu$, Bob announces whether he received outcome $\emptyset$ or not, and both reveal their basis choices.
Furthermore, if $x = y = \Z$, they announce their values $a$ and $b$.
Formally, denoting by $U_i = (\mu_i, x_i, a_i)$ and $V_i = (y_i, b_i)$ Alice's and Bob's classical values in the $i$-th round,
\begin{align*}
\pd(U_i, V_i) = \begin{cases}
(\mu_i, x_i, y_i, a_i, b_i) & \textnormal{if } x_i = y_i = \Z \wedge b_i \neq \emptyset \\
(\mu_i, x_i, y_i, \emptyset) & \textnormal{if } b_i = \emptyset \\
(x_i, y_i) & \textnormal{else.}
\end{cases}
\end{align*}

Alice uses the measurement outcomes from rounds where both her and Bob chose the $X$-basis as the raw key (with $I_i = \pd(U_i, V_i)$): 
\begin{align*}
S_i = \rk(U_i, I_i) = \begin{cases}
a_i & \textnormal{if } x_i = y_i = \X \,, \\
\bot & \textnormal{else.}
\end{cases}
\end{align*}

Finally, to evaluate each round, Bob records the intensity, whether he received outcome $\emptyset$ when Alice sent a photon in the $X$-basis, and whether their outcomes agree in case $x_i = y_i = \Z$.
Formally, 
\begin{align*}
\ev(V_i, I_i, \hat S_i) = \begin{cases}
(\mu_i, \emptyset_{\X}) & \textnormal{if } x_i = \X \wedge \, b_i = \emptyset,\\
(\mu_i, \err_\Z) 
& \textnormal{if } x_i = y_i = \Z \wedge\, a_i \neq b_i \neq \emptyset, \\
\bot & \textnormal{else.}
\end{cases}
\end{align*}
We therefore see that decoy state protocols naturally fit into the framework of \cref{prot:qkd-pm}.

The decoy state BB84 protocol is simple enough to be analysed analytically in the i.i.d.~asymptotic setting.
Therefore, instead of using the numerical technique from \cref{sec:col_attack_bounds}, we can instead reuse these analytical results as our collective attack bound.
Concretely, we need to bound $H(S|IEC)_\nu$ for any collective attack.
For the decoy state BB84 protocol, this analysis can be performed analytically, so we do not need to invoke the numerical technique described in the main text.
We briefly sketch the derivation of the analytical bound and refer to~\cite{ma2005practical,lim2014concise} for details.

We define the ``transmission probability'' $t^x_s$ as the probability that if Alice sends out an $s$-photon state in basis $x \in \zx$, Bob will detect at least one photon, i.e.~not receive outcome $\emptyset$.
Note that $t^x_0 > 0$, i.e.~Bob may detect a photon even though Alice did not send one, either due to a dark count in Bob's detector or due to Eve sending a photon instead.
Similarly, the ``failure probability'' $f^x_s$ is defined as the probability that if Alice sends out $s$ photons in basis $x$, Bob's measurement outcome $b$ will be different from Alice's chosen value $a$, conditioned on Bob not receiving $\emptyset$.
We note that because Alice only knows the intensity of her laser, not how many photons are in a particular pulse, the above quantities are not directly accessible to Alice and Bob.
Also recall that Alice chooses intensity $\mu \in \{\mu_1, \mu_2, \mu_3\}$ with probability $p_{\mu_i}$, and that for a given choice of $\mu$ the number of photons is distributed as
\begin{align*}
p_L(s|\mu) = e^{-\mu} \frac{\mu^s}{s!}\,.
\end{align*}

The security analysis now proceeds in two steps: first, we bound $H(S|IEC)_\nu$ in terms of $t^x_s$ and $f^x_s$.
Then we bound the latter quantities in terms of statistics Alice and Bob can observe in the protocol.
To bound $H(S|IEC)_\nu$, we observe that if Alice did not send any photons but Bob had a detection, the corresponding raw key bit has 1 bit of entropy.
On the other hand, if Alice sent exactly 1 photon, the standard BB84 analysis based on uncertainty relations~\cite{berta2010uncertainty} applies and the corresponding raw key bit has $1 - h(f^\Z_1)$ bits of entropy, where $h$ is the binary entropy function.
Finally, if Alice sends multiple photons, the corresponding raw key bit cannot be guaranteed to have any entropy due to the photon number splitting attack.
Since we consider the photon number to be a classical quantity, we can simply take the weighted average over these cases and find that for any collective attack, 
\begin{align}
H(S|IEC)_\nu \geq \tau_0 t^\X_0 + \tau_1 t^\X_1 (1 - h(f^\Z_1)) \label{eqn:ds_entr} \,,
\end{align}
where $\tau_s$ is the probability that Alice sends out and $s$-photon state: 
\begin{align*}
\tau_s = \sum_{i = 1}^3 p_{\mu_i} p_L(s|\mu_i) \,.
\end{align*}

To bound $t^x_s$ and $f^x_s$, we define the corresponding quantities $t^x_\mu$ and $f^x_\mu$ as the transmission and failure probabilities if Alice sends out a laser pulse with intensity $\mu$.
The $t^x_\mu$ and $f^x_\mu$ are observable quantities because Alice knows which intensity she set for her laser.
Clearly, 
\begin{align*}
t^x_{\mu_i} = \sum_{s=0}^\infty p_L(\mu_i|s) t^x_s \,, \tand 
f^x_{\mu_i} = \sum_{s=0}^\infty p_L(\mu_i|s) f^x_s \,,
\end{align*}
where 
\begin{align*}
p_L(\mu_i|s) = \frac{p_{\mu_i}}{\tau_s} p_L(s|\mu_i) 
\end{align*}
according to Bayes' rule.
From this, one can derive the following bounds by straightforward algebra (see~\cite{ma2005practical,lim2014concise} for details): 
\begin{align*}
t_0^x &\geq \frac{\mu_2 - \mu_3}{\tau_0} \left( \frac{\mu_2 e^{\mu_3} t_{\mu_3}^x}{p_{\mu_3}} - \frac{\mu_3 e^{\mu_2} t_{\mu_2}^x}{p_{\mu_2}} \right) \,,\\
t_1^x &\geq \frac{\mu_1 \tau_1}{\mu_1(\mu_2 - \mu_3) - (\mu_2^2 - \mu_3^2)} \left( \frac{e^{\mu_2} t^x_{\mu_2}}{p_{\mu_2}} - \frac{e^{\mu_3} t^x_{\mu_3}}{p_{\mu_3}} + \frac{\mu_2^2 - \mu_3^2}{\mu_1^2}\left( \frac{t^x_0}{\tau_0} - \frac{e^{\mu_1}t^x_{\mu_1}}{p_{\mu_1}} \right) \right)\,,\\
f_1^x &\leq \frac{\tau_1}{\mu_2 - \mu_3} \left( \frac{e^{\mu_2} f^x_{\mu_2}}{p_{\mu_2}} - \frac{e^{\mu_3} f^x_{\mu_3}}{p_{\mu_3}} \right) \,. \numberthis \label{eqn:ds_lower_bds}
\end{align*}

We now note that $t^\X_{\mu_i}$ and $f^\Z_{\mu_i}$ are related to the statistics stored in register $C$ by
\begin{align*}
t^\X_{\mu_i} = 1 - \pr{C = (\mu_i, \emptyset_\X)} \,, \quad f^\Z_{\mu_i} = \pr{C = (\mu_i, \err_\Z)} \,.
\end{align*}
Therefore, if we insert \cref{eqn:ds_lower_bds} into \cref{eqn:ds_entr}, we obtain a collective attack bound $\ca(\nu_C)$.
To use \cref{thm:general_qkd_pm}, one can take any affine lower bound to this function.
\end{appendices}

\end{document}